\documentclass[draftclsnofoot,onecolumn]{IEEEtran}
\usepackage{amsmath,amsthm,amssymb,hyperref}
\usepackage{amsfonts,dsfont}
\usepackage{verbatim}
\usepackage[pdftex]{graphicx}
\usepackage[utf8]{inputenc}
\usepackage[english]{babel}
\usepackage[font={footnotesize}]{caption}
\usepackage{subcaption}
\usepackage{enumitem}
\usepackage{cite}
\usepackage{verbatim}
\usepackage{mathtools}

\usepackage{color}
\title{Paper Title}
\author{Borna. K }
\date{\today}

\DeclarePairedDelimiter{\ceil}{\lceil}{\rceil}
\DeclarePairedDelimiter{\floor}{\lfloor}{\rfloor}

\newtheorem{theorem}{Theorem}
\newtheorem{corollary}{Corollary}
\newtheorem{lemma}{Lemma}
\newtheorem{proposition}{Proposition}
\newcounter{nmdthmcnt}

\usepackage{chngcntr}
\newcounter{pretheorem}
\counterwithin{theorem}{pretheorem}
\newtheorem{Theorem}[pretheorem]{Theorem}

\theoremstyle{definition}
\newtheorem{condition}{Condition}
\newtheorem{remark}{Remark}
\newtheorem{definition}{Definition}
\newtheorem{claim}{Claim}

\newcommand{\theoremgroup}{\refstepcounter{pretheorem}}

\newcommand{\vm}[1]{\textbf{#1}}

\newcommand{\h}{\mathsf{T}}
\newcommand{\SNR}{\mathsf{SNR}}
\newcommand{\DoF}{\mathsf{DoF}}
\newcommand{\el}{\ell}

\makeatletter
\@ifpackageloaded{hyperref}%
  {\newcommand{\mylabel}[2]
    {\protected@write\@auxout{}{\string\newlabel{#1}{{#2}{\thepage}%
      {\@currentlabelname}{\@currentHref}{}}}}}%
  {\newcommand{\mylabel}[2]
    {\protected@write\@auxout{}{\string\newlabel{#1}{{#2}{\thepage}}}}}
\makeatother

\begin{document}

\title{
$K$--User Interference Channel with Backhaul Cooperation: DoF vs. Backhaul Load Trade--Off
}
\author{
\IEEEauthorblockN{Borna Kananian, \IEEEauthorrefmark{1},
Mohammad~A.~Maddah-Ali, \IEEEauthorrefmark{2},
Babak~H.~Khalaj, \IEEEauthorrefmark{1}
}

\IEEEauthorblockA{
\IEEEauthorrefmark{1}Department of Electrical Engineering, Sharif University of Technology, Tehran, Iran\\
Email: borna@ee.sharif.edu,~khalaj@sharif.edu
}

\IEEEauthorblockA{
\IEEEauthorrefmark{2}Nokia Bell Labs, New Jersey, USA, Email: mohammad.maddah-ali@nokia.com
}
}

\maketitle

\begin{abstract}
In this paper, we consider multiple-antenna $K$--user interference channels with backhaul collaboration in
one side (among the transmitters or among the receivers) and investigate the
trade--off between the rate in the channel versus the communication load in the
backhaul. 
In this investigation, we focus on a first order approximation result, where the rate of the wireless channel is measured by the degrees of freedom (DoF) per user, and the load of the backhaul is measured by the entropy of backhaul messages per user normalized by $\log$ of transmit power, at high power regimes.
This trade--off is fully characterized for the case of even values of $K$, and approximately characterized for the case of odd values of $K$, with vanishing approximation gap as $K$ grows. 
For full DoF, this result establishes the optimality (approximately) of the most straightforward scheme, called \emph{Centralized Scheme}, in which the messages are collected at one of the nodes, centrally processed, and forwarded back to each node.
In addition, this result shows that the gain of the schemes, relying on distributed processing, through pairwise communication among the nodes (e.g., cooperative alignment) does not scale with the size of
the network.
For the converse, we develop a new outer-bound on the trade--off based on splitting the set of collaborative nodes (transmitters or receivers) into two subsets, and assuming full cooperation within each group. 
In continue, we further investigate the trade--off for the cases, where the backhaul or the wireless links (interference channel) are not fully connected. 
\end{abstract}
\begin{IEEEkeywords}
Interference Alignment, Cooperation Alignment, Degrees of Freedom, Centralized Processing.
\end{IEEEkeywords}

\section{Introduction}
Interference is known as the major limiting factor in the performance of wireless communication. 
Theoretically, techniques such as \textit{interference alignment}~\cite{cadambe08,motahari14,maddah10_com} promise to increase the throughput of the network significantly to its first order. 
However, implementation of such techniques faces serious practical challenges.
Fortunately, in a major class of wireless networks, i.e.  cellular networks, there exist some backhaul links, which provide the possibility of collaboration among  the interfering links.
Such backhaul resources can be used to manage interference and increase throughput in wireless links.  
The major question here is how much rate improvement is expected, for a given increase in the backhaul load? 
 
 In~\cite{wang11a},  the authors investigate the effect of receivers cooperation in a two--user interference channel, and characterize the capacity region versus backhaul load trade-off within a constant gap. 
 The achievable scheme is based on a form of Han and Kobayashi method.
 In high $\SNR$ regimes, the scheme of \cite{wang11a} reduces to a simple strategy.
 In particular, whenever cooperation is allowed,  it is optimum to pass the received signal (indeed its quantized version) at receiver one to receiver two, perform the joint decoding at receiver two, and pass the decoded message of receiver one back again. 
For the case of no cooperation, orthogonal transmission (e.g. in time or frequency) is used.

Let $\DoF^*$ and $\alpha$ respectively denote, the optimum rate in wireless link and the backhaul load per user, normalized by $\log (P)$, for large transmission power $P$.
Then from \cite{wang11a}, we conclude that
\begin{align}\label{eq:dof-2}
 \DoF = \min\left\{1\ ,\ \frac{1+\alpha}{2}\right\}.
\end{align}
From the results in~\cite{wang11b}, the same trade--off region can also be derived for two--user interference channel with \emph{transmitters} cooperation.

 The result in \cite{wang11a} has been extended in \cite{ashraphijuo2014capacity} to the cases where users are equipped with multiple antennas.
 To be more specific, the authors in \cite{ashraphijuo2014capacity} consider a multiple antenna two--user interference channel with receivers cooperation and provide an approximate capacity region assuming some fixed backhaul capacity.
 It is shown that the gap between the inner and outer bounds is a function of total number of antennas at the receivers and independent of the signal power, therefore, the trade--off between $\DoF^*$ versus $\alpha$ is fully characterized. 

 In~\cite{ntranos2015cooperation}, the authors characterize the trade--off between DoF versus receivers backhaul load for three user single antenna interference channel. 
The achievable scheme proposed in \cite{ntranos2015cooperation} is fundamentally different from the scheme of \cite{wang11a} and it has a new ingredient in the cooperation scheme, called cooperation alignment.
Such approach outperforms the schemes that are based on collecting all received signals at one node and jointly perform the decoding at that node, by fifty percent and achieves the optimum trade--off.
In cooperation alignment, some alignment techniques has been used in developing backhaul messages such that at each receiver the interfering terms in the signals received through the backhaul and through the wireless link are aligned, and these two together provides the means for canceling the interference and revealing the desired signal.
As such alignment is not possible in one shot solutions, the idea in \cite{ntranos2015cooperation} is to divide the intended signals into many sub-signals.
 In such method, upon receiving backhaul messages, each receiver is able to cancel part of the interference and continuing the message passing phase, all the receivers are able to decode their intended signals.
 Similarly, for the case of no cooperation interference alignment is exploited.
In \cite{ntranos2015cooperation} it is shown that $\DoF^*$ versus $\alpha$ trade--off follows the same formula as in~\eqref{eq:dof-2}. 
 
 The results of \cite{wang11a} and \cite{ntranos2015cooperation} suggest that for $K$--user case, the $\DoF^*$ versus $\alpha$ follows~\eqref{eq:dof-2} as well.
 The main objective of this paper is to show that such generalization is in fact wrong.
 Indeed the cases of two and three users interference channels are only exception rather than a rule.
 
\begin{figure}
\centering
   \includegraphics[scale=0.5]{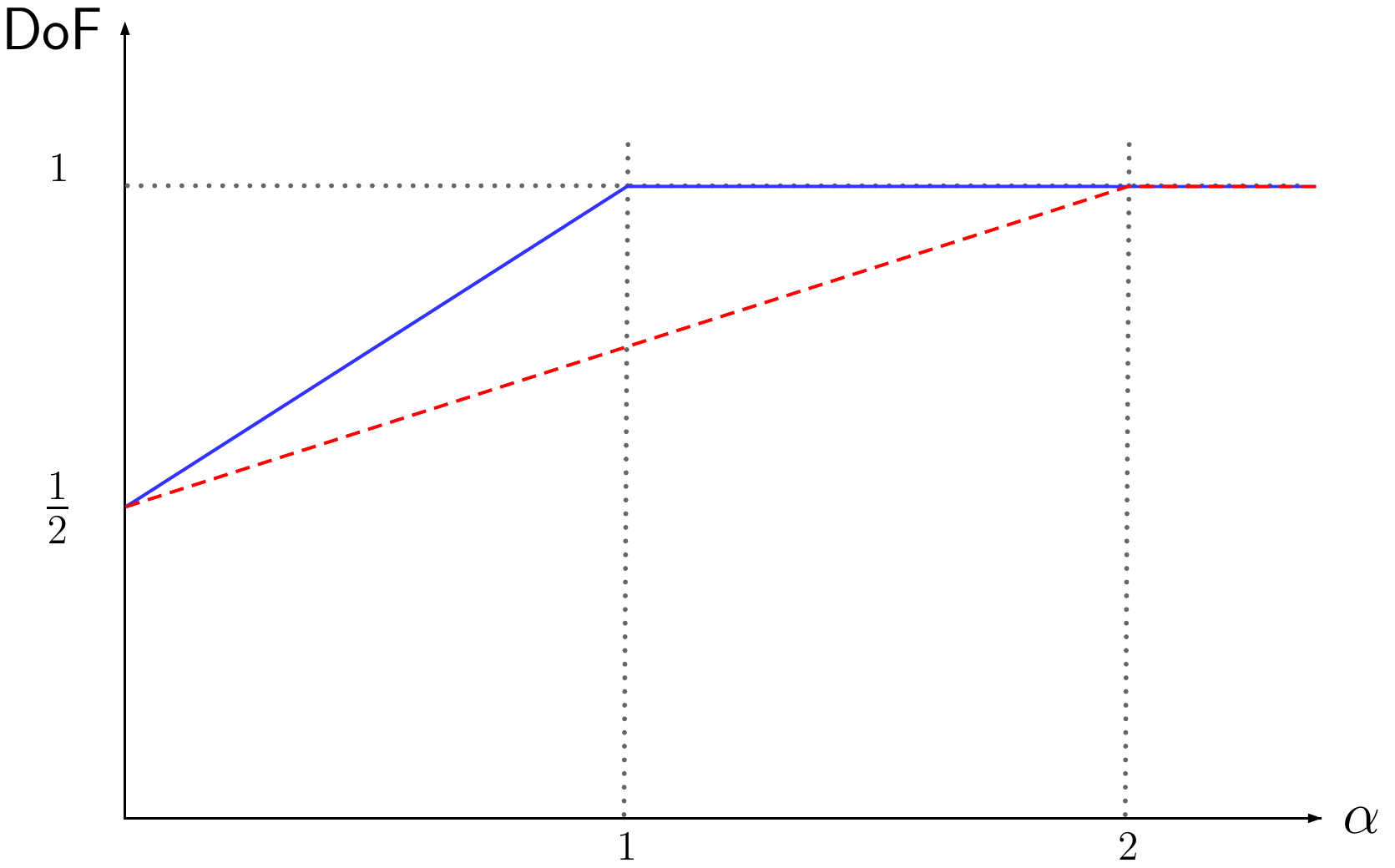}
   \caption{DoF vs. Backhaul load trade--off region.
   The solid blue line shows the region for two and three--user interference channels.
   The dashed red line shows the region for large interference networks.}
   \label{fig:trade-all} 
\end{figure}

 In this paper, we consider a $K$--user multiple antenna interference channel with backhaul cooperation.
 We assume fully connected backhaul network and non-degenerate wireless channels and characterize the full DoF versus backhaul load trade--off region, for both receivers cooperation and transmitters cooperation cases.
 To be more specific, we derive the full trade--off region for the cases of even number of users, while for the cases of odd number of users an achievable bound and a converse bound is provided.
 The gap between the achievable bound and the converse bound in case of odd number of users vanishes as the number of users $K$ increases, hence characteristic of the full trade--off region for large $K$ is presented by
 \begin{align}\label{tradeoff1}
 \DoF^* = \min\left\{M\ ,\ \frac{M}{2}+\frac{\alpha}{4}\right\},
 \end{align}
where $M$ is the number of antennas at each individual node.

When no cooperation is allowed, the achievable scheme is based on interference alignment.
On the other hand to achieve full degrees of freedom per user, we use a centralized scheme which is of great practical interest.
For the case of receivers cooperation, we collect the quantized version of all the received signals at one of the receivers, where the joint decoding takes place.
Then, the decoded signals are sent back to the corresponding receiver through the backhaul network.
For the transmitters cooperation, we collect all the messages in one of the transmitters where an interference management scheme (such as zero forcing) takes place.
Then, the computed signals are sent back to the corresponding transmitters.

In order to develop an outer bound, we propose a new converse based on dividing the set of cooperating nodes into two balanced sets and assuming full cooperation inside each set and only considering the cooperation load between the sets.
For finite number of users, we show that our achievable scheme is optimal when the number of users are optimal.
For odd number of users, we show that the gap between the achievable scheme and the converse is diminishing as the number of users increases.

 In \cite{wang11a,wang11b,ashraphijuo2014capacity,ntranos2015cooperation} the backhaul networks are fully connected, and the wireless channels from each transmitter to every receiver are degenerate with zero probability.
However, these are not valid assumptions for a cellular network, specially when the number of users $K$ increases.
Specifically, due to the power limitations, the signal sent by a transmitter can be detected by only a limited number of receivers. 
In addition, although there exist paths connecting every arbitrary pair of base stations, there are limitations in the backhaul network, and not all the base stations are directly connected to each other.
As an example, in \cite{kananian2016collaboration}, the authors assume that the backhaul network follows linear Wyner model, and each transmitter only interferes on the two receivers closest to it.

Then, unlike \cite{wang11a,wang11b,ashraphijuo2014capacity,ntranos2015cooperation}, we assume the system follows general connectivity.
In general connectivity, at the wireless side, the wireless channel between some of the transmitters and receivers are so low that we can consider them to be zero.
Subsequently at the wireless side, with respect to the channel coefficient matrix between a transmitter and a receiver, either all elements are identically drawn from a continuous probability distribution or all of them are zero.
In addition, in general, at the backhaul side, some of the direct links between cooperating nodes do not exist.

In the second part of this paper, we find conditions on the wireless and the backhaul connectivity, such that the centralized scheme remains optimum.
However, it is important for the conditions to be tractable, since we might deal with very large networks.
We have shown that the proposed conditions are tractable, i.e., verifiable in polynomial time with respect to the network size $K$.
 
 The rest of this paper is organized as follows.
 In Section~\ref{section:formulation}, we formulate our problem from the information theoretic view point and the main results of the paper are presented in Section~\ref{section:main}.
 Section~\ref{sec:two-user} contains the discussions on the two--user interference channels, forming the foundations of the proofs for the main results.
 The proofs of our main results for the case of full connectivity is presented in Section~\ref{sec:proof}, while the discussion on generalized configurations is presented in Section~\ref{section:centralized}.
 Finally, in Section~\ref{sec:com} we present the complexity analysis of the conditions for the optimality of the central processing.
 
 \section{Problem Formulation}\label{section:formulation}
 Consider an \textit{Interference Channel}, with $K$ receivers and a transmitter corresponding to each receiver.
Each transmitter is equipped with $M$ antennas and each receiver is equipped with $N$ antennas.
The set of all transmitters and the set of all receivers are denoted by $\mathcal{T}$ and $\mathcal{R}$, respectively.
Transmitter $i\in\mathcal{T}$, intends to convey a message $W_i$ to its corresponding receiver.
 The channel is a Gaussian Interference Channel, which, in a narrow-band environment, is given by
\begin{align}
\vm{y}_i(t) = \sum_{k=1}^K \vm{H}_{ik} \vm{x}_k(t) + \vm{z}_i(t).
\end{align}
In the above equation, $\vm{y}_i(t)\in \mathbb{C}^{N}$ is the received signal at receiver $i$ and $\vm{x}_k(t)\in \mathbb{C}^M$ is the transmitted signal from transmitter $k$, $\vm{z}_i(t)\in\mathbb{C}^N$ is the additive circularly symmetric Gaussian noise at receiver $i$ with zero mean and \textit{identity} co-variance matrix, all at time $t$.
In addition, $\vm{H}_{ik}\in \mathbb{C}^{N\times M}$ is the channel coefficient matrix from transmitter $k$ to receiver $i$ and is assumed to be fixed during the whole wireless transmission period.
We assume sufficient distance among the antennas, consequently all the channel coefficients are independent, yet randomly chosen from an identical distribution.
We further assume that the full Channel State Information (CSI) is available at all transmitters and receivers.

To be more precise, we assume that the channel coefficient matrix from receiver $i$ to receiver $j$ consists of the multiplication of a \emph{large scale factor} $L_{ij}\in\{0,1\}$ and a small scale factor $\bar{\vm{H}}_{ij} \in \mathbb{C}^{N\times M}$.
The elements of the small scale matrices are independent and randomly chosen from a continuous distribution.
This means  that the channel coefficient matrix $\vm{H}_{ij} = L_{ij}\bar{\vm{H}}_{ij}$ is either zero, or non-zero with probability one.
We collect all the large scale channel coefficients in a binary matrix $\vm{L}$, denoted by the adjacency matrix of the channel. 

\subsubsection*{Representation by a Bipartite Graph}
 The channel has an equivalent bipartite graph $G_{tr}$, with bipartitions $(\mathcal{T},\mathcal{R})$.
In general, there exists a link among receiver $i$ and transmitter $j$ if the large scale channel coefficient $L_{ij}$ is one. 

In most parts of this paper, we assume that the equivalent bipartite graph is fully connected. 
To be more precise, this means that all the receivers are subject to the interference from all the transmitters.
In the other parts, although we assume that small scale coefficients can be zero, we assume that the wireless \emph{direct} channel coefficient matrices are full rank with probability one, as defined rigorously in \nameref{cond2}(Condition~\ref{cond2}).

\begin{condition}[Direct Connectivity Condition]\label{cond2}
In the wireless channels, for every $i\in\mathcal{T}$, $L_{ii}$ is one and hence, the matrix $\vm{H}_{ii}$ is full rank with probability one.
\end{condition} 

The justification for \nameref{cond2}(Condition~\ref{cond2}), is that in the scheduling stage, if a direct link is weaker than a given threshold, we do not assign the receiver to that transmitter.
\vspace{.2 in}

 We assume that there exists a backhaul network providing the ability of cooperation as depicted in Fig.~\ref{fig:sys_model}.
The backhaul network can be either at the receivers side (Fig.~\ref{fig:sys_model1}), or at the transmitters side (Fig.~\ref{fig:sys_model2}) as described in what follows.
 The backhaul network has an equivalent graph $G_b = (\mathcal{V},\mathcal{E})$, where $\mathcal{V}$ is the set of all cooperating nodes and $\mathcal{E}$ represents the set of all available backhaul links.
In most part of this paper, we assume that the graph of the backhual network is fully connected, i.e., there exists a link between every pair $[i,j],\ j\neq i$, of the cooperating nodes through which, they can directly pass backhaul messages.
In the other parts, we assume that some of the direct links in the backhaul network are zero.

\begin{figure}
\centering
\begin{subfigure}[b]{0.7\textwidth}
   \includegraphics[width=1\linewidth]{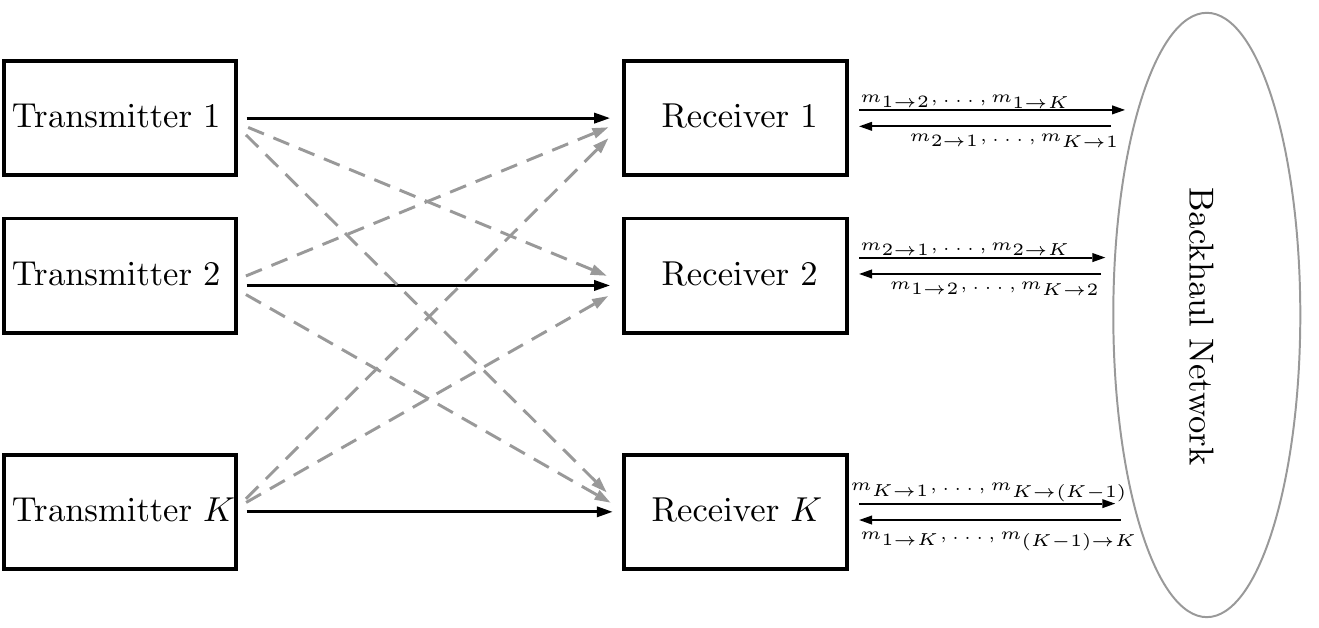}
   \caption{Receivers Cooperation.}
   \label{fig:sys_model1} 
\end{subfigure}

\begin{subfigure}[b]{0.7\textwidth}
   \includegraphics[width=1\linewidth]{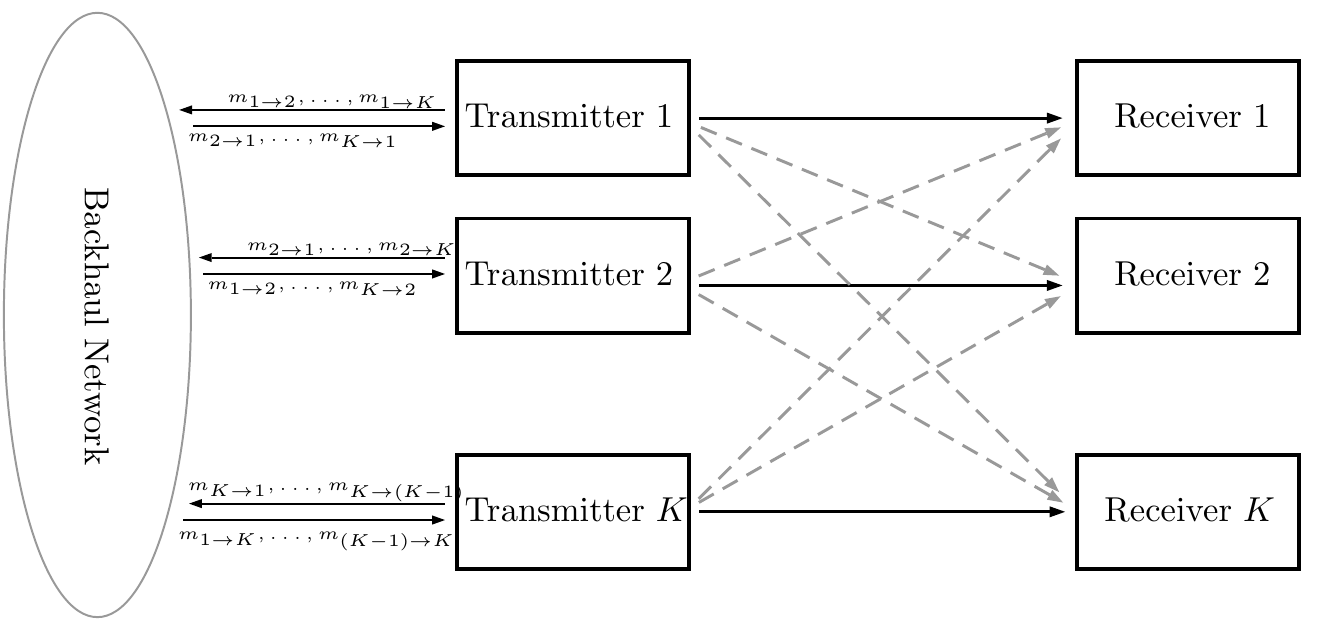}
   \caption{Transmitters Cooperation.}
   \label{fig:sys_model2}
\end{subfigure}

\caption{System Model}
\label{fig:sys_model}
\end{figure}

\subsection{Receivers Cooperation}\label{scenario:receivers}
 In this scenario, receivers cooperate through the backhaul network as depicted in Fig.~\ref{fig:sys_model1}.
Here, the set of vertices in the equivalent graph of the backhaul network is the set of all receivers, i.e., $\mathcal{V}=\mathcal{R}$.
In order for transmitter $i$ to create the signal $\vm{x}_i(t)$, it encodes the message $W_i$ into a codeword $[\vm{x}_i(t)]_{t=1}^n$, using a block code of length $n$ subject to the following average power constraint
\begin{align}\label{eq:ave_pow_const}
\frac{1}{n}\sum_{t=1}^n Tr(\vm{x}_i(t)\vm{x}_i^\dagger(t))\leq P.
\end{align}
Note that the outgoing signals of the transmitters are completely independent, since they are not connected and the messages $W_i$ and $W_j$ are independent, for all $i\neq j$.

Transmitter $i$ and its corresponding receiver agree on the message sets $\mathcal{W}_i = \{1,2,\dots,2^{n R_i}\}$.
It uses the encoding function 
\begin{align*}
f_i^n: \mathcal{W}_i \rightarrow \mathbb{C}^{M\times n},
\end{align*}
where
\begin{align*}
[\vm{x}_i(t)]_{t=1}^n = f_i^n(W_i).
\end{align*}
Subsequently, the receivers participate in the message passing process.
The backhaul message $m_{i\rightarrow j}(t)$, sent from receiver $i$ to receiver $j$ at $t^{th}$ round of collaboration phase, is chosen from $\mathcal{B}_{ij}^t$ which is a finite set denoting the backhaul message alphabet.
In order to construct the message $m_{i\rightarrow j}(t)$, receiver $i$ employs all the signals previously received over its wireless terminal $[y_i(\tau)]_{\tau=1}^{t-1}$, as well as all previously received backhaul messages $[m_{j\rightarrow i}(\tau)]_{\tau=1}^{t-1},\ \forall j\in\{1,2,\dots,K\}/i$.

In particular, receiver $i$ uses the backhaul message generating function
\begin{align*}
g_{ij}^{[t]}: \underset{k\in\{1,\dots,K\}\backslash i, \tau \in\{1,\dots,t-1\}}{ \mathbb{C}^{N\times (t-1)}\ \times\quad \prod \mathcal{B}_{ki}^\tau \rightarrow }\mathcal{B}_{ij}^t\ ,\ j\neq i,
\end{align*}  
  in order to form a backhaul message to receiver $j$ at $t^{th}$ round of collaboration phase.
Let $M_i^{[t]}$ be the collection of all received backhaul messages upto $t^{th}$ round of collaboration phase, at receiver $i$, i.e.,
\begin{align}\label{eq:previous_messages}
M_i^{[t]}= \big\{[m_{k\rightarrow i}]_{\tau=1}^{t}, k\in\{1,2,\dots,K\}\backslash i\big\},
\end{align}
then, we have
\begin{align*}
m_{i\rightarrow j}(t) = g_{ij}^{[t]} ([y_i(\tau)]_{\tau=1}^{t-1}, M_i^{[t-1]}).
\end{align*}

 Finally, for receivers to decode their intended signals, receiver $i$ chooses the decoding function 
\begin{align*}
 \eta_i : \underset{k\in\{1,\dots,K\}\backslash i, t\in\{1,\dots,n\}}{ \mathbb{C}^{N\times n}\ \times\quad \prod \mathcal{B}_{ki}^t\rightarrow }\mathcal{W}_i,
\end{align*} 
in order to decode its desired message, where
\begin{align*}
\hat{W}_i =  \eta_i([\vm{y}_i(\tau)]_{\tau=1}^{n}, M_i^{[n]}).
\end{align*} 

The corresponding probability of error can  be calculated as 
\begin{align}\label{eq:err_prob}
P_e^{(n)} = \mathbb{P}\big(\cup_{i\in\{1,\dots,K\}}\{ \hat{W}_i \neq W_i \}\big).
\end{align} 

\subsection{Transmitters Cooperation}\label{scenario:transmitter}
 As shown in Fig.~\ref{fig:sys_model2}, in this scenario, transmitters cooperate through a backhaul network.
 Note that in this scenario, the set of vertices in the equivalent graph of the backhaul network is the set of all transmitters, i.e., $\mathcal{V}=\mathcal{T}$.
In order for transmitter $i$ to create the signal $\vm{x}_i(t)$, it exploits the message $W_i$ and all the received backhaul messages at the cooperation phase.
The resulting signal $[\vm{x}_i(t)]_{t=1}^n$ contains $n$ time slots and is subject to the average power constraint \eqref{eq:ave_pow_const}.

In this case, transmitter $i$ and its corresponding receiver, agree on the message set $\mathcal{W}_i = \{1,2,\dots,2^{n R_i}\}$.
Then the transmitters participate in the message passing process.
The backhaul message $m_{i\rightarrow j}(t)$, sent from transmitter $i$ to transmitter $j$ at $t^{th}$ round of collaboration phase, is chosen from $\mathcal{B}_{ij}^t$ which is a finite set denoting the backhaul message alphabet.
In order to construct the message $m_{i\rightarrow j}(t)$, transmitter $i$ uses its intended message $W_i$, as well as all previously received backhaul messages $[m_{j\rightarrow i}(\tau)]_{\tau=1}^{t-1},\ \forall j\in\{1,2,\dots,K\}/i$.

In particular, transmitter $i$ uses the backhaul message generating function 
\begin{align*}
g_{ij}^{[t]}: \underset{k\in\{1,\dots,K\}\backslash i, \tau \in\{1,\dots,t-1\}}{\mathcal{W}_i\ \times\quad \prod \mathcal{B}_{ki}^\tau \rightarrow }\mathcal{B}_{ij}^t\ ,\ j\neq i,
\end{align*}
to form a backhaul message to transmitter $j$ at $t^{th}$ round of collaboration phase, i.e.,
\begin{align*}
m_{i\rightarrow j}(t) = g_{ij}^{[t]} (W_i, M_i^{[t-1]}),
\end{align*}
where $M_i^{[t]}$ is defined at \eqref{eq:previous_messages}.

In addition, transmitter $i$ uses an encoding function 
\begin{align*}
f_i^n:\underset{k\in\{1,\dots,K\}\backslash i, t\in\{1,\dots,n\}}{ \mathcal{W}_i\quad\times\quad \prod \mathcal{B}_{ki}^t\rightarrow }\mathbb{C}^{M\times n},
\end{align*}
in order to convey its message to the intended receiver through the interference channel
\begin{align*}
[x_i]_{t=1}^{n} = f_{i}^{n} (W_i, M_i^{n}).
\end{align*}

Finally, receiver $i$ uses a decoding function 
\begin{align*}
\eta_i :  \mathbb{C}^{N\times n}\rightarrow \mathcal{W}_i,
\end{align*}
in order to decode  its desired message, i.e., 
\begin{align*}
\hat{W}_i = \eta_i([\vm{y}_i(\tau)]_{\tau=1}^{n}),
\end{align*}
where $\hat{W}_i$ is the decoded message at receiver $i$. 
The corresponding probability of error can  be calculated accordingly, as in \eqref{eq:err_prob}. 

\subsection{Capacity and DoF Region}
We define the rate of each backhaul link and the average (per user) cooperation rate, as in \cite{ntranos2015cooperation}.
Specifically, the rate of each backhaul link is defined as the average entropy of the messages passing through that link
\begin{align}\label{eq:backhaul_rate}
R_b^{[i,j]} = \frac{1}{n} H([m_{i\rightarrow j}(\tau)]_{\tau=1}^n),
\end{align}	
and the average cooperation rate is also defined as the sum rate of all backhaul links, normalized by the number of users
\begin{align}
\bar{R}_b = \frac{1}{K} \sum_{i=1}^K \sum_{j\neq i} R_b^{[i,j]}.
\end{align}

For the achievablity of the rate vector $\vm{R}=[R_1, R_2, \ldots, R_k]^\h$, it is required that for every $\epsilon\ \textgreater\ 0$, there exists an integer $n_0$ such that $P_e^{(n)}\ \textless\ \epsilon$, for every  block length $n\ \textgreater\ n_0$, while 
\begin{align}\label{eq:ave_back_const}
\bar{R}_{b} = \frac{1}{K} \sum_{i=1}^K\sum_{j\neq i} \frac{1}{n}H([m_{i\rightarrow j }(\tau)]_{\tau=1}^n)\leq L.
\end{align}
Here, $L$ is a function of the average power constraint and indicates the capacity of the backhaul network.

The closure of all achievable rate vectors, subject to the average backhaul constraint \eqref{eq:ave_back_const}, forms the capacity region $\mathcal{C}_L$. 
Considering only the interference effect on the capacity region, we omit the noise impact by solely focusing on high $\SNR$ regimes.

In high $\SNR$ regimes, we define the backhaul capacity of the link from cooperating nodes $i$ to cooperating node $j$ as,
\begin{align*}
c_B^{ij} \triangleq \lim_{P\rightarrow\infty} \frac{R_b^{[i,j]}(P)}{\log(P)},
\end{align*}
and the average (per user) backhaul cooperation load as,
\begin{align*}
\alpha \triangleq \lim_{P\rightarrow\infty} \frac{L(P)}{\log(P)}.
\end{align*}
 To be more specific, we define the backhaul capacity and the average backhaul load as the limit for the backhaul rate and the capacity of the backhaul network divided by the rate of a point to point Gaussian channel at high $\SNR$, respectively.
 
In the same way, we also define 
 the achievable degrees of freedom for each individual user as,
\begin{align*}
\DoF_i(\alpha) \triangleq \liminf_{P\rightarrow\infty} \frac{R_i}{\log(P)},i\in\{1,\dots,K\}.
\end{align*}
and the average (per user) achievable degrees of freedom as,
\begin{align*}
\DoF(\alpha) \triangleq \liminf_{P\rightarrow\infty} \frac{1}{K} \sum_{k=1}^K \frac{R_k}{\log(P)}.
\end{align*}

The average \textit{DoF} (per user) of the channel is denoted by $\DoF^*(\alpha)$ and is defined as the supremum of $\DoF(\alpha)$, over all achievable schemes.
We denote $\DoF^*(\alpha)$ as the $\DoF$ region, and characterizing the trade-off between the $\DoF^*$ and the backhaul load $\alpha$, is one of the main focus points of this paper.

Definition~\ref{def:centralized_scheme} precisely defines a class of centralized schemes.
This definition generalized the concept of centralized scheme and provides the proper tool for writing rigorous converse proofs.
\begin{definition}\label{def:centralized_scheme}
The class of $(\alpha,\DoF)$--centralized scheme consists of all the schemes, achieving degrees of freedom of $\DoF$ per user with the backhaul load of $\alpha$ per user, where each receiver $k$ decodes its own message $\hat{W}_k$ and at least one of the cooperating nodes, say cooperating node $i$, is able to decode $\tilde{W}_1$, $\tilde{W}_2$, $\ldots$, $\tilde{W}_K$ with vanishing probability of error.
This means that for every $\epsilon\ \textgreater\ 0$, there exists an integer $n_0$ that for codes with the block lengths $n\geq n_0$, we have
$\underset{k}\max\ \mathbb{P}(\hat{W}_k \neq W_k) \leq \epsilon$, and $\underset{k}\max\ \mathbb{P}(\tilde{W}_k \neq W_k) \leq \epsilon$.
We denote cooperating node $i$ as the \emph{Central Processor}.
\end{definition}

\begin{definition}\label{def:fisibility}
For a given configuration, we say that the class of $(\alpha,\DoF)$--centralized schemes is feasible, if at least one of its members is achievable under \nameref{cond2}(Condition~\ref{cond2}).
\end{definition}

\section{Main Results}\label{section:main}

In this section, we present our main results.
For the rest of the paper we consider the cases where all transmitters and receivers have equal number of antennas, i.e., $M=N$, unless otherwise stated.

\subsection{Fully Connected Wireless Network}\label{sec:full}
In this subsection, we bound the achievable DoF per user with limited backhaul capacity, where both the wireless and the backhaul networks are fully connected.
This means that at the backhaul network, for every arbitrary pair of cooperating nodes, there exists a two-way backhaul link. 
Also, in the wireless channel, all the entries of the adjacency matrix are one, i.e., the channel coefficients matrix $\vm{H}_{ij}$, for all $i\in\mathcal{R}$ and $j\in\mathcal{T}$, are full rank.
We treat the problem separately for the cases of \textit{Even} and \textit{Odd} number of users. 

\theoremgroup\label{theorem:bound_receiver}
\begin{theorem}\label{theorem:bound_receiver_e}
In a $K$--user interference channel where each node is equipped with $M$ antennas, for even values of $K$ and with the average  backhaul load $\alpha$, we have
\begin{align}\label{eq:dof_receiver_b_e}
\DoF^*(\alpha)= \min \{M,\frac{1}{2} (M+\frac{K}{2(K-1)}\alpha)\}.
\end{align}
\end{theorem}

\begin{theorem}\label{theorem:bound_receiver_o}
In a $K$--user interference channel where each node is equipped with $M$ antennas, for odd values of $K$ and with the average  backhaul load $\alpha$, we have
\begin{align}\label{eq:dof_receiver_b_o}
\min \{M,\frac{1}{2} (M+\frac{K}{2(K-1)}\alpha)\}
\leq \DoF^*(\alpha)
\leq \min \{M,\frac{K+1}{2K} (M+\frac{\alpha}{2})\}.
\end{align}
\end{theorem}

\begin{remark}
The results of Theorem~\ref{theorem:bound_receiver}, are valid for the case of receivers cooperation (\ref{scenario:receivers}) as well as the case of transmitters cooperation (\ref{scenario:transmitter}).
Therefore, we have shown the reciprocity of DoF vs. Backhaul load trade--off for both receivers and transmitters cooperation.
This was previously discussed in \cite{wang11a} and \cite{wang11b} for two--user interference channel for rate versus backhaul load trade-off.
Here, we extend the result to $K$--user interference channels in terms of DoF versus backhaul load.
To be more specific, we show that the effect of transmitters cooperation and receivers cooperation on the DoF are similar and the gain from exploiting either of them is the same.
\end{remark}

\begin{remark}\label{remark:achievable}  For the achievability, we use time-sharing between two corner points. 
In case of no collaboration (i.e., $\alpha=0$), we use interference alignment to achieve DoF of $\frac{M}{2}$ per user. 
 On the other hand, to eliminate the entire effect of interference, and approximately achieving the capacity of $K$ interference-free MIMO links, we follow \emph{centralized processing} as follows:
\begin{itemize} 
\item{Receivers Cooperation}:
A quantized version of the received signals at all the receivers are collected at one receiver.
 At that receiver, the decoding is done jointly, and the decoded messages are sent back to the corresponding receivers.
\item{Transmitters Cooperation:} The messages corresponding to all the transmitters are collected at one transmitter.
That transmitter encodes each of the messages, and the encoded signals then go through to a linear transformation, by multiplying to the inverse of the channel coefficient matrix (i.e., zero forcing).
The resulting signals are sent back to the corresponding transmitter, and subsequently, wireless transmission phase takes place.
Thus, each receiver receives the interference free version of the encoded message of its corresponding transmitter.
 \end{itemize}
\end{remark}

\begin{remark} Using the results of Theorem~\ref{theorem:bound_receiver}, we settled down the problem raised in \cite{ntranos2015cooperation}.
To be more specific, we solve the problem of characterizing the DoF vs. backhaul load trade-off region for $K$--user interference channel with receiver backhaul cooperation.
In case of \textit{Even} number of users, \eqref{eq:dof_receiver_b_e} characterizes full trade-off region.
Moreover in case of \textit{Odd} number of users, with finite number of antennas, the gap characterized in \eqref{eq:dof_receiver_b_o} vanishes for large network size, $K$.
\end{remark}

\begin{remark} The authors in \cite{ntranos2015cooperation} show that in case of three--user Interference channel with receivers cooperation and single antenna users, it is possible to reduce the amount of required backhaul load by employing the concept of \emph{Cooperation Alignment} and decentralized encoding and decoding.
Note that the problem formulation admits the solutions in which interference management is performed through pairwise collaboration.
However, our results show that the most straightforward scheme in which one node collects all the signals, perform interference management, and send back the computed signals to the corresponding nodes is (almost) optimal.  
\end{remark}

\begin{remark} According to \eqref{eq:dof_receiver_b_o}, the minimum required backhaul load for $\DoF^*(\alpha)=M$,  is given by $\alpha_{\min} = \frac{2M(K-1)}{K+1}$.
This bound is tight for the case of three--user interference channel with single antenna users, which is consistent with the result of \cite{ntranos2015cooperation}.
It is worth noting that the difference between $\alpha_{\min}$ and the required backhaul load in the Centralized Scheme (briefly discussed in Remark~\ref{remark:achievable}) is $\frac{2M(K-1)}{K(K+1)}$, which for finite number of antennas shrinks as $K$ increases.
\end{remark}

\begin{corollary}\label{cor:01}
For the full DoF vs. backhaul load trade-off region of a $K$--user Interference Channel with backhaul cooperation, we have
\begin{align*}
\DoF^*(\alpha)= \min \{M,\frac{1}{2} (M+\frac{\alpha}{2})\},
\end{align*}
for large enough number of users $K$.
\end{corollary}

\begin{remark}\label{remark:01}
Consider the scenario in which we have $K$, receivers each equipped with $M$ antennas.
Corresponding to each receiver there exist $M$ single antenna transmitters, i.e., $MK$ total transmitters. 
In this scenario, the results of Theorem~\ref{theorem:bound_receiver} hold. 
\end{remark}

\begin{remark}\label{remark:02}
Consider another scenario in which we have $K$ transmitters each equipped with $M$ antennas.
Corresponding to each transmitter there exist $M$ single antenna receivers, i.e., $MK$ total receivers. 
The results of Theorem~\ref{theorem:bound_receiver} holds in this scenario. 
\end{remark}

The proofs for Remarks~\ref{remark:01} and~\ref{remark:02} follow the same line of proof of Theorem~\ref{theorem:bound_receiver} and hence are omitted.

\subsection{General Wireless Networks}\label{sec:general}
In this subsection, we introduce a condition on wireless network connectivity such that in the presence of that condition, the class of $(2M,M)$--centralized schemes are optimum, for large values of $K$.
In a general wireless network, the elements of the adjacency matrix are allowed to become zero, i.e., the channel coefficient matrix $\vm{H}_{ij}$, are either zero or full rank, for all $i\in\mathcal{R}$ and $j\in\mathcal{T}$.
Here, first we introduce Extended Hall Condition as follows.
\begin{condition}[Extended Hall's Condition]\label{cond:ehc}
Let $\ell$ be any arbitrary integer in $\{1,2,\ldots,\ceil{\frac{K}{2}}\}$.
In the equivalent bipartite graph, for each group of $\ell$ arbitrary subset $\mathcal{S}\subset\mathcal{T}$ of transmitters, we have
\begin{align*}
|\mathcal{N}_\mathcal{R}(\mathcal{S})|\ \geq \floor{\frac{K}{2}}+\ell,
\end{align*}
where $\mathcal{N}_\mathcal{R}(\mathcal{S})$ is defined according to Definition~\ref{def:neigh} in Appendix~\ref{app:graph}.
\end{condition}
The reason for the appellation of the above condition is its similarity to the \nameref{cond:hall-cond} (Condition~\ref{cond:hall-cond}).
We have the following theorems.
\begin{Theorem}\label{theorem:condition}
Consider a $K$-user interference channel with backhaul cooperation where the \nameref{cond2} (Condition~\ref{cond2}) holds and all the receivers and transmitters are equipped with $M$ antennas.
Whenever \nameref{cond:ehc} (Condition~\ref{cond:ehc}) holds, the class of $(2M-M)$--centralized schemes is optimum, for large values of $K$. 
\end{Theorem}

\begin{remark}
Centralized scheme is of great interest from the practical point of view.
The above theorem deals with the cases where some of the wireless links between transmitters and receivers are zero.
In this case, the question is whether the DoF of $M$ and the backhaul of $2M$ form a point on the boundray of the region of backhual load versus DoF trade--off.
The above theorem states some necessary condition for such optimality.
\end{remark}

\begin{remark}
Note that \nameref{cond:ehc} (Condition~\ref{cond:ehc}) deals with exponentially many subset of nodes.
To be more precise, in order to check \nameref{cond:ehc} exhaustively is exponentially hard in network size.
Still, Theorem~\ref{theorem:poly} shows that similar to the \nameref{cond:hall-cond}, we can verify \nameref{cond:ehc} in polynomial time in our problem setting.
\end{remark}

\begin{Theorem}\label{theorem:poly}
In a $K$--user interference channel, \nameref{cond:ehc} (Condition~\ref{cond:ehc}) can be verified in a time polynomial in the network size $K$.
\end{Theorem}

\begin{remark}
Together, Theorems~\ref{theorem:condition} and~\ref{theorem:poly} show that, in a specific configuration of the wireless channel, one is able to check whether the centralized scheme is optimal in a reasonable time. 
\end{remark}

\begin{remark}
In the proof of Theorem~\ref{theorem:poly} (Section~\ref{sec:com}), we show that the above problem reduces to \emph{Max-Flow Min-Cut} problem which is solvable in polynomial time.
\end{remark}

In the following section, we focus on two--user interference channel which forms the foundation of the proofs for Theorem~\ref{theorem:bound_receiver}.

\section{Two--User Interference Channels}\label{sec:two-user}
In this section, we focus on the case of two--user interference channels with backhaul cooperation.
The problem setup is similar to Section~\ref{section:formulation}, except we assume that transmitter $i$ is equipped with $M_i$ antennas and receiver $i$ is equipped with $N_i$ antennas.
Specifically, in this section we remove the assumption of equal number of antennas both for the transmitters and the receivers.

Lemma~\ref{lemma:MIMO_dof_bound} characterizes the DoF region of the general two--user interference channel. 
This DoF regions forms the main ingredient for the proof of Theorem~\ref{theorem:bound_receiver} (Section~\ref{sec:proof}), as well as proof of Theorem~\ref{theorem:two-user} (to be discussed later in this section).

\begin{lemma}\label{lemma:MIMO_dof_bound}
For a two--user interference channel with backhaul cooperation, the DoF region is characterized by
\begin{align}
\DoF_1 + \DoF_2 \leq \min\{N_1,(M_1-N_2)^+ \} + \min\{N_2, M_1+M_2 \} + c_B^{12},\label{eq:MIMO_dof_bound_1}\\
\DoF_1 + \DoF_2 \leq \min\{N_2,(M_2-N_1)^+ \} + \min\{N_1, M_1+M_2 \} + c_B^{21}.\label{eq:MIMO_dof_bound_2}
\end{align}
This result holds for receivers cooperation scenario as well as transmitters cooperation scenario.
\end{lemma}

\begin{remark}
The authors in \cite{ashraphijuo2014capacity} have derived the same result (Theorem 2) for the case of receivers cooperation.
However, here we prove the bounds for both transmitters cooperation and receivers cooperation.
The challenge to prove Lemma~\ref{lemma:MIMO_dof_bound} for transmitters cooperation case is that the transmitting signals are not independent, in spite of the fact that the messages are independent.
This is due to the fact that the cooperation phase has taken place before the wireless transmission phase and the transmitters are, to some extend, aware of each others' messages.
Therefore, the method used in \cite{ashraphijuo2014capacity} is not applicable anymore.
Our main contribution here, is to deal with such dependency and showing the reciprocity of transmitters and receivers cooperation case for general two--user interference channel with backhaul cooperation.
To handle this case, we bound the effect of dependency between transmitting messages and show that in terms of degrees of freedom this does not affect the results.
On the other hand, in \cite{ashraphijuo2014capacity} the power constraint is of the form $\mathbb{E}\{\text{Tr}(\vm{x}_i \vm{x}_i^*)\}\leq P$, which is different from \eqref{eq:ave_pow_const}.
This will also affect the proof, since \eqref{eq:ave_pow_const} is more general.
\end{remark}

\begin{remark}
It is worth mentioning that our achievable scheme is completely different of the one in~\cite{ashraphijuo2014capacity}.
The proposed achievable scheme  is based on the centralized scheme, while in~\cite{ashraphijuo2014capacity} a Han--Kobayashi achievable scheme has been proposed.
Therefore, the proposed achievable scheme is considerably simpler.
\end{remark}

\begin{remark}\label{remark:two-user-constraint}
Following the lines of the proof, one is able to verify that the only constraints for the results of Lemma~\ref{lemma:MIMO_dof_bound} to hold, is that the channels coefficient matrices from transmitter $i$ to receivers $i$ and $j$ must be full rank, for $i=1,2$ and $j\neq i$.
This means that the direct and cross channel coefficient matrices must be full rank.
Such observation is of great importance for the discussions in Section~\ref{sec:general}. 
\end{remark}

Using the results from Lemma \ref{lemma:MIMO_dof_bound}, we have the following Theorem for two--user interference channels.  

\begin{Theorem}\label{theorem:two-user}
For a two--user interference channel with backhaul cooperation and average per user backhaul load $\alpha$, where transmitter $i$ is equipped with $M_i$ antennas and receiver $i$ is equipped with $N_i = M_i$ antennas, we have 
\begin{align}
\DoF^*(\alpha) = \frac{1}{2}\min\bigg\{ M_1 + M_2\ ,\ \alpha + \max(M_1,M_2) \bigg\}.
\end{align} 
\end{Theorem}

\begin{proof}
To prove Theorem~\ref{theorem:two-user}, we first present the converse then the achievable scheme.
Without loss of generality, we assume that $M_1 \leq M_2$.
For the converse proof, we directly use the results of Lemma~\ref{lemma:MIMO_dof_bound}, where by setting $N_i = M_i$ we have
\begin{align*}
\begin{split}
\DoF_1 + \DoF_2 \leq \min\{M_1,(M_1-M_2)^+ \} + \min\{M_2, M_1+M_2 \} + c_B^{12},\\
\DoF_1 + \DoF_2 \leq \min\{M_2,(M_2-M_1)^+ \} + \min\{M_1, M_1+M_2 \} + c_B^{21}.
\end{split}
\end{align*}
By adding the above two inequalities, we have
\begin{align*}
\begin{split}
\DoF_1 + \DoF_2 \leq \frac{1}{2} \big( (M_2-M_1) +  (M_1+M_2) + c_B^{12} + c_B^{21} \big),
\end{split}
\end{align*}
or
\begin{align*}
\begin{split}
\DoF (\alpha) \leq \frac{ M_2 + \alpha }{2} .
\end{split}
\end{align*}
On the other hand, since there are $M_i$ antennas at each transmitter $i$ and each receiver $i$, the maximum achievable DoF equals $\frac{M_1 + M_2}{2}$ per user. 
This completes the converse proof.

For the achievable scheme, we use the time sharing between two corner points.
The first corner point is where no cooperation is allowed and we simply turn the firss transmitter--receiver pair off and let the second pair work to achieve a DoF of $\frac{M_2}{2}$ per user.
Moreover, the second corner point is obtained by eliminating the entire effect of interference, and approximately achieve DoF of $\frac{M_1 + M_2}{2}$ per user.
We use the central processing as introduced in Remark~\ref{remark:achievable} while the cooperating node two handles the processes.
This scheme requires $2M_1$ backhaul messages to achieve DoF of $\frac{M_1+M_2}{2}$ and the backhaul load is $M_1$ per user.
\end{proof}

\begin{remark}
For a two--user interference channel where each transmitter and each receiver are equipped with $M$ antennas, the DoF versus backhaul load trade of region is given by
\begin{align*}
\DoF^*(\alpha) = \min \big\{M, \frac{M+\alpha}{2}\big\},
\end{align*}
which is consistent with the results of Theorem~\ref{theorem:bound_receiver}.
\end{remark}

So far we have dealt with two--user multiple antenna interference channels.
Now we are ready to move to more general case and prove Theorem~\ref{theorem:bound_receiver}.

\section{Proof of Theorem~\ref{theorem:bound_receiver}}\label{sec:proof}
In this section, we first provide a converse proof for Theorem~\ref{theorem:bound_receiver}, and then we prove the achievability of the theorem.

\subsection{Converse Proof}\label{subsection:converse}\label{subsection:converse}
 First, we notice that since all the transmitters and all the receivers are equipped with $M$ antennas, we have $\DoF_i \leq M$, for all $i$, and consequently $\DoF \leq M$.
We call this as the single user bound.

For the rest of the proof we develop an upper--bound using the DoF region of a two--user interference channel with backhaul cooperation as stated in Lemma~\ref{lemma:MIMO_dof_bound}.
 
Let us partition the whole set of transmitter--receiver pairs, into two groups.
In this partitioning, there exists $K_1$ pairs of transmitters-receivers in the first group and $K_2=K-K_1$ pairs in the second group.
Note that such a partitioning is not unique and by putting different pairs of transmitter--receiver in different groups, we can form different realizations of such a partitioning.
Let $\mathcal{K}_i$ represent the set of indices of transmitter-receiver pairs in group $i$, and $\mathcal{K}$ denote the set of indices of all transmitter-receiver pairs.
For any $\mathcal{S}\subseteq \mathcal{K}$ we define 
\begin{align}
\DoF^*_\mathcal{S} \triangleq \sum_{i\in \mathcal{S}} \DoF^*_i.
\end{align}

Corresponding to each realization of the above partitioning, we form a two--user interference channel, in which the transmitters (receivers) in $\mathcal{K}_i$ fully cooperate with each other at zero backhaul cost.
In other words, the transmitters (receivers) in $\mathcal{K}_i$ form the multiple antenna transmitter $i$ (receiver $i$).
We also set the capacity of the backhaul link connecting cooperating node $i$ to cooperating node $j$ in the two--user interference channel equal to the sum of the capacities of the backhaul links connecting each of the cooperating nodes in group $i$ to each of the cooperating nodes in group $j$ of the original $K$--user interference channel (For an example see Fig.~\ref{fig:MIMO}).
\begin{align}\label{eq:backhaul_cap}
\begin{split}
\hat{c}_B^{[\mathcal{K}_1,\mathcal{K}_2]} \triangleq \sum_{i\in\mathcal{K}_1,j\in\mathcal{K}_2} c_B^{ij}, \\
\hat{c}_B^{[\mathcal{K}_2,\mathcal{K}_1]} \triangleq \sum_{i\in\mathcal{K}_2,j\in\mathcal{K}_1} c_B^{ij}.
\end{split}
\end{align}

\begin{figure}[t!]
\centering
\includegraphics[scale=0.9]{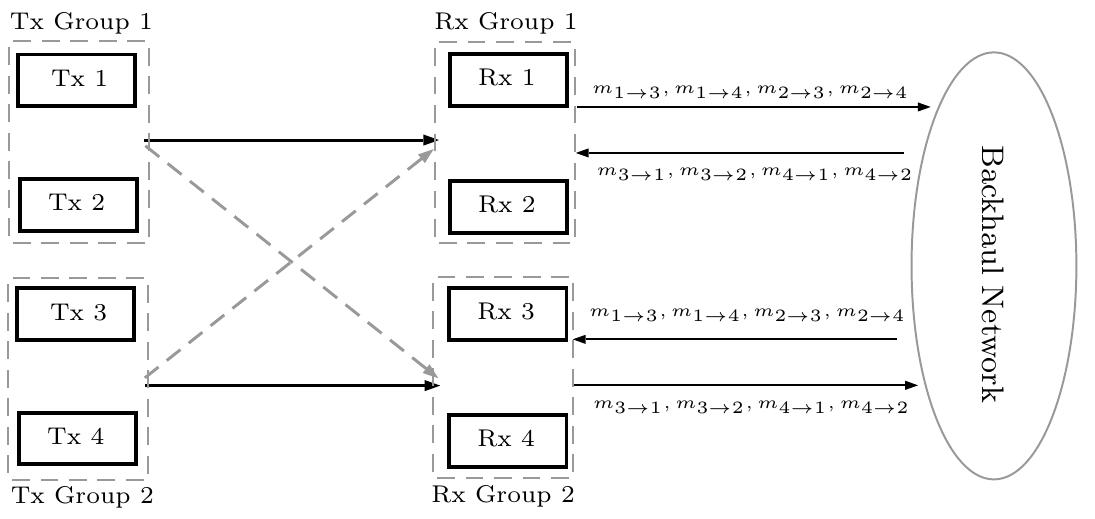}
\caption{This is a special realization of users partitioning. In this realization we partition receivers one and two to the first group of receivers and correspondingly, transmitters one and two to the first group of transmitters.
Permitting cost-free intra-group collaboration, the system is equivalent to a two--user two-antenna interference channel.
Receivers groups are able to collaborate through the backhaul network.}
\label{fig:MIMO}
\end{figure}

\begin{lemma}\label{claim:upper_bound}
For a given backhaul load limit, the total achievable DoF in each realization of the above partitioning, is upper bounded by the total achievable DoF of the two--user interference channel corresponding to that realization.
\end{lemma}

\begin{proof}
Any achievable solution on the original interference channel with backhaul cooperation can be used on the corresponding two--user interference channel.
\end{proof}

According to Lemma~\ref{lemma:MIMO_dof_bound} and~\ref{claim:upper_bound}, one concludes that for any given realization of the partitioning, we have
\begin{align}\label{eq:tow_user_bound}
\begin{split}
\DoF^*_{\mathcal{K}_1} + \DoF^*_{\mathcal{K}_2} \leq M(K_1 - K_2)^+
+M K_2+ \hat{c}_B^{[\mathcal{K}_1,\mathcal{K}_2]},\\
\DoF^*_{\mathcal{K}_1} + \DoF^*_{\mathcal{K}_2} \leq M(K_2 - K_1)^+ +M K_1 + \hat{c}_B^{[\mathcal{K}_2,\mathcal{K}_1]},
\end{split}
\end{align}
where $M$ is the number of antennas for each individual transmitter (receiver) and $\hat{c}_B^{ij}$ is defined according to~\eqref{eq:backhaul_cap}.
In addition, $\mathcal{K} = \mathcal{K}_1\cup\mathcal{K}_2$ and $\mathcal{K}_1\cap\mathcal{K}_2 = \emptyset$, therefore $\DoF^*_{\mathcal{K}_1} + \DoF^*_{\mathcal{K}_2}$ equals the total achievable DoF, i.e.,
\begin{align*}
\sum_{k=1}^K \DoF_i = \DoF^*_{\mathcal{K}_1} + \DoF^*_{\mathcal{K}_2}.
\end{align*}

Here, we consider the proof for Theorems~\ref{theorem:bound_receiver_e} and~\ref{theorem:bound_receiver_o} separately.
\vspace{0.2 in}

\noindent
\textbf{Part I}\ (Converse for Theorem~\ref{theorem:bound_receiver_e}.)

In case of even number of users, we have $K=2m$ for some $m\in\mathbb{N}$. 
In this case we choose $K_1=K_2=m$.
By substituting $K_1$ and $K_2$ and adding both equations in \eqref{eq:tow_user_bound} we have
\begin{align}\label{eq:bound_e}
\DoF^*_{\mathcal{K}_1} + \DoF^*_{\mathcal{K}_2} \leq  Mm + \frac{\hat{c}_B^{[\mathcal{K}_1,\mathcal{K}_2]} +\hat{c}_B^{[\mathcal{K}_2,\mathcal{K}_1]}}{2},
\end{align}
for all realizations of $\mathcal{K}_1$ and $\mathcal{K}_2$ with $|\mathcal{K}_1|=|\mathcal{K}_2|=m$.

In order to form one of such realizations, we have ${K}\choose{m}$ different options.
Consider the realizations in which, pair $i$ is in one group and pair $j$ is in the other, for $i,j\in\mathcal{K}$ and $i\neq j$.
In such realizations, the capacity of the backhaul link connecting node $i$ to node $j$, denoted by $c_B^{ij}$, and the capacity of the backhaul link connecting node $j$ to node $i$ denoted by $c_B^{ji}$, show up in \eqref{eq:bound_e}.
In order to form a realization, we need to choose $m-1$ pairs, out of $K-2$ remaining pairs.
This means that there are ${K-2}\choose{m-1}$ different options. 

Using the above argument and summing \eqref{eq:bound_e} over all options, we have
\begin{align*}
\frac{1}{2} {{K}\choose{m}} \big(\sum_{i\in\mathcal{K}} \DoF^*_i - Mm\big) \leq \frac{1}{2} {{K-2}\choose{m-1}}\underset{i,j\in\mathcal{K},j\neq i}\sum c_B^{ij}.
\end{align*}
Thus
\begin{align*}
\frac{K!}{m!m!}\big(\sum_{i\in\mathcal{K}} \DoF^*_i - Mm\big) &\leq \frac{(K-2)!}{(m-1)!(m-1)!} \underset{i,j\in\mathcal{K},j\neq i}\sum c_B^{ij},
\end{align*}
therefore 
\begin{align*}
\sum_{i\in\mathcal{K}} \DoF^*_i &\leq m(M +\frac{1}{2(K-1)} \underset{i,j\in\mathcal{K},j\neq i}\sum c_B^{ij}).
\end{align*}
Since $m = \frac{K}{2}$, we have
\begin{align}\label{eq:converse-even}
\DoF^* = \frac{1}{K} \sum_{i\in\mathcal{K}} \DoF^*_i &\leq \frac{1}{2} \big(M + \frac{K}{2(K-1)} \alpha\big).
\end{align}

\noindent
\textbf{Part II}\ (Converse for Theorem~\ref{theorem:bound_receiver_o}.)

In case of odd number of users, we have $K=2m+1$, for some $m\in\mathbb{N}$. 
In this case we choose $K_1=m$ and $K_2=m+1$.
By substituting $K_1$ and $K_2$ and adding both equations in \eqref{eq:tow_user_bound} we have
\begin{align}\label{eq:bound_o}
\DoF^*_{\mathcal{K}_1} + \DoF^*_{\mathcal{K}_2} \leq M(m+1) + \frac{\hat{c}_B^{[\mathcal{K}_1,\mathcal{K}_2]}+\hat{c}_B^{[\mathcal{K}_1,\mathcal{K}_2]}}{2}.
\end{align}
for any realization of $\mathcal{K}_1$ and $\mathcal{K}_2$ with $|\mathcal{K}_1|=m$ and $|\mathcal{K}_2| = m+1$.

In order to form one of such realizations, we need to choose $m$ pairs, out of $K$ available pairs.
This means that there are ${K}\choose{m}$ different options.
Consider a realization in which, pair $i$ is in one group and pair $j$ is in the other, for some $i,j\in\mathcal{K}$ and $i\neq j$.
In this realization, the capacity of the backhaul link connecting node $i$ to node $j$, denoted by $c_B^{ij}$, and the capacity of the backhaul link connecting node $j$ to node $i$, denoted by $c_B^{ji}$, show up in \eqref{eq:bound_o}.
To form such a realization, we need to choose $m$ pairs, out of $K-2$ remaining pairs.

Using the above argument and summing \eqref{eq:bound_o} over all realizations we have
\begin{align*}
 {{K}\choose{m}} \big(\sum_{i\in\mathcal{K}} \DoF^*_i - M(m+1)\big) &\leq  {{K-2}\choose{m}} \underset{i,j\in\mathcal{K},j\neq i}\sum c_B^{ij}.
\end{align*}
Thus
\begin{align*}
 \frac{K!}{m!(m+1)!} \big(\sum_{i\in\mathcal{K}} \DoF^*_i - M(m+1)\big) &\leq  \frac{(K-2)!}{m!(m-1)!}\underset{i,j\in\mathcal{K},j\neq i}\sum c_B^{ij}.
\end{align*}
Consequently,
\begin{align*}
\sum_{i\in\mathcal{K}} \DoF^*_i &\leq  (m+1) (M + \frac{1}{2K}\underset{i,j\in\mathcal{K},j\neq i}\sum c_B^{ij}).
\end{align*}
Since $K=2m+1$, we have,
\begin{align}\label{eq:converse-odd}
\DoF^* = \frac{1}{K} \sum_{i\in\mathcal{K}} \DoF^*_i \leq \frac{K+1}{2K} \big(M + \frac{\alpha}{2} \big),
\end{align}
which completes the proof.

\begin{remark}\label{remark:multi-user-constraint}
Referring to Remark~\ref{remark:two-user-constraint} and the fact that the converse proof is based on the results from two--user interference channel, one is able to verify that for the converse bounds \eqref{eq:converse-even} and \eqref{eq:converse-odd} to hold, it is required that for all realizations the channel coefficient matrices from the transmitters in group $i$ to receiver in group $i$ and $j$ must be full rank, for $i=1,2$ and $j\neq i$.
More details are given in Lemma~\ref{lemma:direct-cross}.
\end{remark}

\subsection{Achievable Schemes}\label{subsection:achievable}

First consider the case where no cooperation is permitted, i.e., $\alpha = 0$.
This case is widely investigated in the literature and it is shown that a DoF of  half per user is achievable using \textit{Interference Alignment} \cite{cadambe08,motahari14,maddah10_com}.
Consequently, we have, 
\begin{align*}
\DoF^*(0) \geq \frac{M}{2}
\end{align*}

 Let us now focus on the cases where each user is able to achieve full DoF of $M$.
We treat this part separately for the cases of receivers cooperation and transmitters cooperation.
 
\subsubsection{Receivers Cooperation}\label{sub:rec}
In this scenario, transmitter $i$ uses a Gaussian code book, carrying $M$ degrees of freedom and transmit the message $\vm{X}_i^n = [\vm{x}_i(1),\vm{x}_i(2),\ldots,\vm{x}_i(n)]$.
At the receivers side, one of the receivers, say receiver 1, is chosen to be the central processor. 
Then, all other receivers, quantize their signals with unit squared error distortion and send it to the central processor, i.e. receiver 1, using the backhal links.

 The central processor, upon receiving all the backhal messages, is able to jointly process all received signals and decode all the messages. 
Exploiting the backhaul links again, central processor sends back the desired message of each receiver. 

 \subsubsection{Transmitters Cooperation}\label{sub:tra}
In this case each transmitter uses a Gaussian code book, carrying $M$ degrees of freedom.
The cooperation phase takes place prior to transmission and one of the transmitters, say transmitter 1, is chosen be to the central processor.
Then, all other transmitters send their intended messages to the central processor using the backhaul links.
 Upon receiving all the backhaul messages, the central processor is aware of all intended messages and it encodes each message separately for the corresponding receiver.
 With the encoded messages and the channels knowledge in hand, the central processor performs a linear beamforming.
To be more specific, let $\bar{\vm{H}}\in\mathbb{C}^{MK\times MK}$ denote the supper channel matrix, and $\vm{x}_j(\el)$ denote the signal intended for receiver $j$ at time $\el$. 
The central processor performs a linear zero--forcing and forms the signal $\vm{u} = \bar{\vm{H}}^{-1}\vm{x}$.
Here, $\vm{x} = [\vm{x}_1(\el)^\h,\vm{x}_2(\el)^\h,\ldots,\vm{x}_K(\el)^\h]^\h$ and $\vm{u} = [\vm{u}_1^\h,\vm{u}_2^\h,\ldots,\vm{u}_K^\h]^\h$ and $|\vm{u}_k| = |\vm{x}_k|$ for $k\in\{1,2,\ldots,K\}$.

Subsequently, it quantizes the signal $\vm{u}$ with unit squared error distortion, resulting in $\bar{\vm{u}}$, and for all $k\in\{1,2,\ldots,K\}$, and sends $\bar{\vm{u}}_k$ back to transmitter $k$ using the backhal links.
Each transmitter consequently sends what it has received from the central processor through the channel.
After the transmission phase, each receiver is able to decode its desired signal, since the received signals are interference-free.

Each of the above achievable schemes, so-called the \textit{Centralized Scheme}, requires $\alpha = \frac{2M(K-1)}{K}$ per user backhaul load and achieves full DoF of $M$ per user, i.e.,
\begin{align*}
(\DoF^*)^{-1}(M) \leq \frac{2M(K-1)}{K}
\end{align*}

Up to now we have shown the achievability of two corner points of the DoF versus backhaul load trade-off region, i.e., $(\frac{M}{2},0)$ and $(M,\frac{2M(K-1)}{K})$.
Using time-sharing, one is able to achieve every point connecting these two points in the trade-off region. 
Consequently, we have
\begin{align*}
\DoF^*(\alpha)\geq \min\{M,\frac{1}{2}(M+\frac{K}{2(K-1)}\alpha)\}.
\end{align*}

\noindent Referring to Definition~\ref{def:centralized_scheme}, both schemes introduced in Sub--sections~\ref{sub:rec} and~\ref{sub:tra} are in class of $(2M,M)$--centralized schemes, for large values of $K$.

\section{Generalized Configurations}\label{section:centralized}
Up to now, we assumed that the backhaul network is fully connected and all the large scale wireless channel coefficients are equal to one.
To be more precise, in our problem setting, the equivalent graph of the backhaul network contains all possible links and the transmitted signal of each transmitter contributes to the signal received by all the receivers.

In the schemes introduced in Section~\ref{subsection:achievable}, we do not use all the links of the backhaul network. 
To be more specific, there are $\frac{K(K-1)}{2}$ links available in a fully connected backhaul network, while our scheme only employs $K-1$ of them.
Therefore, even in the absence of some of the backhaul links the class of $(M,2M)$--centralized scheme might be achievable.

On the other hand, with respect to the wireless links, there are $K^2$ large scale channel coefficients, all of which are one in the fully connected scenario.
However, for the optimality of class of $(M,2M)$--centralized schemes, it is not required for all the large scale channel coefficients to be equal to one.

Thus, we expect the class of centralized schemes to be optimum in some networks with less connectivity, as well.
In this section, we aim to explore such class of networks.

In general configurations, in the backhaul side, some of the direct links among the cooperating nodes do not exist and on the wireless side, the channel between some of the transmitters and receivers are so weak that we can consider them to be zero.
To be more precise, we assume that the large scale channel coefficients are either zero or one.
This means  that the channel coefficient matrix $\vm{H}_{ij} = L_{ij}\bar{\vm{H}}_{ij}$ is either zero ($\vm{H}_{ij} =\boldsymbol{0}$), or $\vm{H}_{ij} = \bar{\vm{H}}_{ij}$ where all the elements of $\bar{\vm{H}}_{ij}$ are drawn i.i.d from a continuous distribution.
Following the definition of the adjacency matrix in Section~\ref{section:formulation}, for a group of transmitters $\mathcal{S}$ and a group of receivers $\mathcal{Q}$, the matrix $\vm{L}_{QS}$ is called the adjacency matrix of the set of transmitters and the set of receivers.

We assume that, \nameref{cond2} (Condition~\ref{cond2}) always holds, i.e., $L_{ii}=1$, for all $i\in\{1,2,\ldots,K\}$, while $L_{ij}$ could be either zero or one for $i\neq j$.
Then, the question is whether the centralized scheme is still optimum in those configurations or not?

The objective in this section is to investigate the optimality and feasibility of the class of centralized schemes in more general configurations.

\subsection{Optimality Condition}\label{sec:optimality}
As mentioned in Corollary~\ref{cor:01}, in a $K$--user interference channel with fully connected backhaul network and wireless channel, the centralized scheme is optimal for large enough values of $K$.
To be more specific, there is no scheme that achieves DoF of $M$ per user, with a backhaul load less than $2M$ per user.

In this sub-section, we assume the backhaul network remains fully connected, and identify some wireless configurations, where despite the fact that zero channels are allowed, the class of $(2M,M)$--centralized schemes, remains to be optimal.
To identify this class of wireless networks, we find some sufficient condition on wireless channel connectivity such that the converse proof in Section~\ref{subsection:converse} is still valid.

First, we have the following theorem about the connectivity of the wireless channel.
\begin{Theorem}\label{theorem:equivalency}
In the presence of \nameref{cond2} (Condition~\ref{cond2}), the following conditions are equivalent:\\
\emph{(a) \nameref{cond:ehc} (Condition~\ref{cond:ehc})}\\
Let $\ell$ be any arbitrary integer in $\{1,2,\ldots,\ceil{\frac{K}{2}}\}$.
In the equivalent bipartite graph, for each group of $\ell$ arbitrary subset $\mathcal{S}\subset\mathcal{T}$ of transmitters,
\begin{align*}
|N_\mathcal{R}(\mathcal{S})|\ \geq \floor{\frac{K}{2}}+\ell.
\end{align*}
\emph{(b) Full Rank Sub-Matrices Conditions}\\
Let $\ell_1$, $\ell_2\in\{\ceil{\frac{K}{2}},\floor{\frac{K}{2}}\}$ be such that $\ell_1 + \ell_2 \geq K$.
For any arbitrary subset $\mathcal{S}\subseteq\mathcal{T}$ of transmitters with $|\mathcal{S}| = \ell_1$ and any arbitrary subset $\mathcal{Q}\subseteq\mathcal{R}$ of receivers with $|\mathcal{Q}| = \ell_2$, $\vm{H}_{\mathcal{Q}\mathcal{S}}$ is full rank, almost surely.
\vspace{0.1in}

\noindent
\emph{(c) Non Degenerate Direct and Cross Channels Conditions}\\
Let  $\ell \in\{\ceil{\frac{K}{2}}, \floor{\frac{K}{2}}\}$.
For each arbitrary set $\mathcal{S}$ of users with $|\mathcal{S}|=\ell$, the channel coefficient matrices $\vm{H}_{\mathcal{S}\mathcal{S}}$ and $\vm{H}_{\mathcal{Q}\mathcal{S}}$ are full rank with probability one, where $\mathcal{Q} = \mathcal{R}\backslash \mathcal{S}$.
\end{Theorem}

\begin{proof}

We first show the equivalency of \emph{(a)} and \emph{(b)}.

\emph{a $\Longleftarrow$ b}
\\
Assume \emph{(a)} does not hold.
Then, there exists at least one $\ell\in\{1,2,\ldots,\ceil{\frac{K}{2}}\}$ and a set $\mathcal{S}\subset \mathcal{T}$, containing $\ell$ transmitters, with less than $\floor{\frac{K}{2}}+l$ neighboring receivers, i.e.,
$
|\mathcal{N}_\mathcal{R}(\mathcal{S})|\ \textless\ \floor{\frac{K}{2}}+\ell
$.
Moreover, there exists a set $\mathcal{Q}\subset \mathcal{R}$, containing $\ceil{\frac{K}{2}}$ receivers, and 
$|\mathcal{Q} \cap \mathcal{N}_\mathcal{R}(\mathcal{S})|\ \textless\ \ell$,
i.e., $
|\mathcal{N}_\mathcal{Q}(\mathcal{S})|\ \textless\ \ell$.
Consider another set $\bar{\mathcal{S}}\subset \mathcal{T}$, containing $\ceil{\frac{K}{2}}$ transmitters, where $\mathcal{S}\subseteq \bar{\mathcal{S}}$.

There exists a subset $\mathcal{S}$ of $\bar{\mathcal{S}}$ with $\ell$ nodes, which has less than $\ell$ neighbors in $\mathcal{Q}$.
According to Theorem~\ref{theorem:hall}, there is no perfect matching between $\bar{\mathcal{S}}$ and $\mathcal{Q}$.
This results in rank-deficiency of the coefficient matrix of the channel between $\bar{\mathcal{S}}$ and $\mathcal{Q}$, due to Proposition~\ref{prop:2}.
Therefore, \emph{(b)} does not hold.

\emph{a $\Longrightarrow$ b}
\\
We provide the proof, for the cases of even and odd number of users separately.\\
\textit{Case 1:} Even number of users ($K=2m$)

Here, it is only required to show that the channel coefficient matrix from any set $\mathcal{S}\subset \mathcal{T}$ containing $m$ transmitters to any set $\mathcal{Q}\subset \mathcal{R}$ containing $m$ receivers, is full rank.

Assume \emph{(a)} holds, i.e., for any $\ell\in\{1,2,\ldots,m\}$ and any arbitrary set $\mathcal{S}\subset \mathcal{T}$, containing $\ell$ transmitters, there are at least $m+\ell$ neighboring receivers.
Since there are $2m$ receivers in total, one concludes that in each arbitrary set $\mathcal{Q}\subset \mathcal{R}$, containing $m$ receivers, i.e., $|\mathcal{Q}|=m$, there are at least $\ell$ neighbors  of $\mathcal{S}$, i.e.,
$
|\mathcal{N}_\mathcal{Q}(\mathcal{S})| \geq \ell$.
According to Theorem~\ref{theorem:hall}, a perfect matching exists for each sub-graph of $(\mathcal{S},\mathcal{Q})$ of the equivalent bipartite graph, where $|\mathcal{S}|=|\mathcal{Q}|=m$.
Exploiting Proposition~\ref{prop:2}, the existence of the perfect matching results in the rank-efficiency of the equivalent matrix, i.e., \emph{(b)} holds.\\
\textit{Case 2:} Odd number of users ($K=2m+1$)

Here, it is only required to show that the channel coefficient matrix from any set $\mathcal{S}\subset \mathcal{T}$ containing either $m$ or $m+1$ transmitters to any set $\mathcal{Q}\subset \mathcal{R}$ containing $m+1$ receivers, is full rank.
It is also required to show that the channel coefficient matrix form any set $\mathcal{S}\subset \mathcal{T}$ containing $m+1$ transmitters to any set $\mathcal{Q}\subset \mathcal{R}$ containing either $m$ or $m+1$ receivers, is full rank.

Assume \emph{(a)} holds, i.e., for any $\ell\in\{1,2,\ldots,m+1\}$ and any set $\mathcal{S}\subset \mathcal{T}$, containing $\ell$ transmitters, there are at least $m+\ell$ neighboring receivers.
Since there are $2m+1$ receivers in total, one concludes that in each arbitrary set $\mathcal{Q}\subset \mathcal{R}$, containing $m+1$ receivers, i.e., $|\mathcal{Q}|=m+1$, there are at least $\ell$ neighbors of $\mathcal{S}$, i.e.,
$|\mathcal{N}_\mathcal{Q}(\mathcal{S})| \geq \ell.
$. 
Therefore, and according to Theorem~\ref{theorem:hall}, for each sub-graph of $(\mathcal{S},\mathcal{Q})$ where $|\mathcal{S}|\leq m+1$ and $|\mathcal{Q}|=m+1$ there exists a \emph{matching} saturating all transmitters.

On the other hand, since \emph{(a)} holds and according to Lemma~\ref{lemma:trans_rec}, for any $k\in\{1,2,\ldots,m+1\}$ and any set $\mathcal{Q}\subset \mathcal{R}$, containing $k$ \textbf{receivers}, there are at least $m+k$ neighboring \textbf{transmitters}.
With the same argument as above, since there are totally $2m+1$ transmitters, one concludes that in each arbitrary set $\mathcal{S}\subset \mathcal{T}$, containing  $m+1$ transmitters, i.e., $|\mathcal{S}|=m+1$, there are at least $k$ neighbors of $\mathcal{Q}$, i.e., $
|\mathcal{N}_\mathcal{S}(\mathcal{Q})| \geq k$.
Therefore, and according to Theorem~\ref{theorem:hall}, for each sub-graph of $(\mathcal{S},\mathcal{Q})$ where $|\mathcal{S}|=m+1$ and $|\mathcal{Q}|\leq m+1$ there exists a \emph{matching} saturating all \textbf{receivers}.

Exploiting Proposition~\ref{prop:2}, the existence of a matching results in the rank-efficiency of the equivalent matrix, i.e., \emph{(b)} holds.

Now we provide the proof for the equivalency of \emph{(b)} and \emph{(c)}.
However, it is very easy to verify that if \emph{(b)} holds, then \emph{(c)} also holds.

\emph{b $\Longleftarrow$ c}
\\
Note that, since \nameref{cond2}(Condition~\ref{cond2}) holds, according to Proposition~\ref{prop:1}, the direct channel coefficient matrices are always full rank.
For the rest of proof we use contradiction.
Consider a set $\mathcal{S}\subset \mathcal{T}$, containing $\ceil{\frac{K}{2}}$ transmitters and another set $\mathcal{Q}\subset \mathcal{R}$, containing $\ceil{\frac{K}{2}}$ receivers, with rank--deficient channel coefficients matrix.
According to Proposition~\ref{prop:2}, there is no perfect matching in the adjacency graph and consequently, due to the results from Theorem~\ref{theorem:hall}, there exists a set $\bar{\mathcal{S}}\subseteq \mathcal{S}$, containing $\ell \leq \ceil{\frac{K}{2}}$ transmitters with $|\mathcal{N}_\mathcal{Q}(\bar{\mathcal{S}})|\ \textless\  \ell $.

We define, $\bar{\mathcal{Q}} \triangleq \mathcal{Q}\backslash \mathcal{N}_\mathcal{Q}(\bar{\mathcal{S}})$, and hence, $\ceil{\frac{K}{2}} -\ell\ \textless\ |\bar{\mathcal{Q}}|\ \leq  \ceil{\frac{K}{2}}$ and $|\mathcal{N}_{\bar{\mathcal{Q}}}(\bar{\mathcal{S}})|=0$.
We also define $\hat{\mathcal{S}}$, as the set of receivers corresponding to the transmitters in $\bar{\mathcal{S}}$.
Hence, $|\hat{\mathcal{S}}| = \ell$ and since \nameref{cond2} (Condition~\ref{cond2}) holds, $|\mathcal{N}_{\hat{\mathcal{S}}}(\bar{\mathcal{S}})|=\ell$.
Note that $|\bar{\mathcal{Q}}\cap\hat{\mathcal{S}}|=0$, and $\ceil{\frac{K}{2}}\ \textless\ |\bar{\mathcal{Q}}\cup\hat{\mathcal{S}}|\leq \ceil{\frac{K}{2}}+\ell$.

Since there are totally $K$ receivers, we are able to choose $\floor{\frac{K}{2}}- \ell$ receivers from the set $\mathcal{R}\backslash (\bar{\mathcal{Q}}\cup\hat{\mathcal{S}})$ and put them in set $\mathcal{P}$.
We also define $\bar{\mathcal{P}}$ as the set of transmitters corresponding to the receivers in $\mathcal{P}$.
The set $\hat{\mathcal{P}} = \bar{\mathcal{S}}\cup\bar{\mathcal{P}}$ contains $\floor{\frac{K}{2}}$ transmitters and all the receivers in $\bar{\mathcal{Q}}$ are in the group of receivers whose transmitters are not among these transmitters.
Therefore, and since $\ceil{\frac{K}{2}} - \ell\ \textless\ |\bar{\mathcal{Q}}|$ there exists a subset  $\bar{\mathcal{S}}$ of the set $\hat{\mathcal{P}}$ containing $\ell$ transmitters, with less than $\ell$ neighboring receivers in the non-corresponding receivers set.
Hence, due to Theorem~\ref{theorem:hall} there is no perfect matching and accordingly due to Proposition~\ref{prop:2} the cross channel coefficient matrix is not full rank, i.e., \emph{(c)} does not hold.

\end{proof}

Considering the converse proof in Section~\ref{subsection:converse}, one is able to verify that the converse bounds in~\eqref{eq:converse-even} and~\eqref{eq:converse-odd} remain valid, if Condition~\emph{(c)} in Theorem~\ref{theorem:equivalency} holds (see Remarks~\ref{remark:two-user-constraint} and \ref{remark:multi-user-constraint} for more explanations).
The following lemma provides a more precise statement.

\begin{lemma}\label{lemma:direct-cross}
In $K$--user interference channel with backhaul cooperation where all transmitters and receivers are equipped with a single antenna, if Non Degenerate Direct and Cross Channels Condition (Condition~\emph{(c)} in Theorem~\ref{theorem:equivalency}) holds, the class of $(2M-M)$--centralized schemes is optimal, for large enough values of $K$.
\end{lemma}

\begin{proof}
In Subsection~\ref{subsection:converse}, we used the bounds from two--user interference channels, and add them together in order to derive the converse bounds.
As long as the bounds for each of the two--user channels holds, the bounds in~\eqref{eq:converse-even} and~\eqref{eq:converse-odd} remain valid.

According to the proof of Lemma~\ref{lemma:MIMO_dof_bound}, in order for the bounds of the two--user interference channel to be valid, it is only required to have full rank direct and cross channel coefficient matrices.
Note that the bounds must be valid for each two--user interference channel and therefore the direct and cross channel coefficient matrices must be full rank for each realization of the our grouping in Subsection~\ref{subsection:converse}, i.e., Non Degenerate Direct and Cross Channels Condition (Condition~\emph{(c)} in Theorem~\ref{theorem:equivalency}).

\end{proof}

Theorem~\ref{theorem:equivalency} and Lemma~\ref{lemma:direct-cross} leads to the proof of Theorem~\ref{theorem:condition}.
A natural consequent question will arise on  the benefits of \nameref{cond:ehc} (Condition~\ref{cond:ehc}) over Non Degenerate Direct and Cross Channels Condition (Condition~\emph{(c)} in Theorem~\ref{theorem:equivalency}).

Note that, to exhaustively verify Condition~\emph{(c)} in Theorem~\ref{theorem:equivalency}, we need to check the rank of $2{K\choose \floor{\frac{K}{2}}}$ matrices for odd values of $K$ and ${K\choose \frac{K}{2}}$ matrices for even values of $K$.
For large values of $K$, this leads to checking the rank of approximately $\frac{2^K}{\sqrt{K}}$ matrices, which is not feasible.

However, Theorem~\ref{theorem:poly} states that Condition~\ref{cond:ehc} can be verified in polynomial time.
The detailed proof is given in Section~\ref{sec:com}.

\begin{remark}
Although \nameref{cond:ehc} (Condition~\ref{cond:ehc}) gives the necessary and sufficient condition for Non Degenerate Direct and Cross Channels Condition (Condition~\emph{(c)} in Theorem~\ref{theorem:equivalency}), it only gives a sufficient condition for the converse bounds in~\eqref{eq:converse-even} and~\eqref{eq:converse-odd}, and subsequently the optimality of the class of $(2M,M)$--centralized schemes.
This is due to the fact that there might be other ways to derive the same bounds.
\end{remark}

\subsection{Feasibility Condition}\label{sec:feasibility}

This subsection deals with the backhaul connectivity.

Recall that, according to Definitions~\ref{def:centralized_scheme} and~\ref{def:fisibility}, for the case of receivers cooperation, a member of class of $(2M,M)$--centralized schemes is achievable, if each receiver $i \in \{1,2,\ldots, K\}$ can decode the message $W_i$ with vanishing probability of error, and in addition at least one of the receivers can decode the whole set of messages, again with a vanishing probability of error.

On the other hand, for the case of transmitters cooperation, a memeber of class of $(2M,M)$--centralized schemes is achievable, if each receiver $i \in \{1,2,\ldots, K\}$ can decode the message $W_i$ with vinishing probability of error, while at least one of the transmitters having access to the whole set of messages.

Note that the feasibility of the class of $(M,2M)$--centralized schemes is not equivalent to its optimality.
To be more precise, there exist scenarios in which the class of centralized schemes is feasible, but we can achieve the full DoF of $M$, even without any cooperation, e.g., in case of no interference in wireless channels.
In fact, there might exist an achievable scheme which is able to achieve the DoF of $M$ per user with the per user backhaul load less than $2M$.

The following lemma, states the necessary and sufficient conditions for the feasibility of the class of $(2M,M)$--centralized schemes.

\begin{lemma}\label{prop:suf_nec}
In a $K$--user interference channel with backhaul cooperation, with $K$ large, and in the presence of \nameref{cond2} (Condition~\ref{cond2}) and~\nameref{cond:ehc} (\ref{cond:ehc}), the class of $(2M,M)$--centralized schemes is feasible, if and only if in the backhaul graph, there exists at least one node, say $i$, with degree of connectivity, $K-o(K)$, i.e., 
$$\textnormal{Deg} (i) = K-o(K)$$.
\end{lemma}

\begin{proof}
We first focus on achievable scheme which is trivial.
Assume that a node with connectivity degree almost equal to $K$ exists.
Since $\text{Deg}(i) = K - o(K)$, there are only $o(K)$ nodes that are not in the neighborhood of $i$.
Assume we prevent the transmitters corresponding to these $o(K)$ nodes from participating in the wireless transmission, and perform the scheme in Section~\ref{subsection:achievable}, over the remaining $\bar{K} = K-o(K)$ nodes.
We choose node $i$ to be the central processor, resulting in a total DoF of $M\bar{K}$, i.e., DoF of $M$ per user, by exploiting $2M(\bar{K} - 1)$ backhaul messages, i.e., backhaul load of $2M$ per user.

The converse proof, on the other hand, it not trivial.
In order to pursue that, we start with the following claims.
\begin{claim}\label{claim:number_messages}
In a $K$--user interference channel with backhaul cooperation, in the presence of \nameref{cond2} (Condition~\ref{cond2}) and~\nameref{cond:ehc} (Condition~\ref{cond:ehc}), in order to achieve the full DoF of $M$ per user, it is required for each cooperating node to send and receive a backhaul load greater than $M$. 
\end{claim}
\begin{proof}[Proof of Claim~\ref{claim:number_messages}.]
Let us first investigate the case of two--user interference channels.
Under \nameref{cond2} (Condition~\ref{cond2}) and~\nameref{cond:ehc} (Condition~\ref{cond:ehc}), and according to \eqref{eq:MIMO_dof_bound_1}, in a two--user interference channel where the transmitters and the receivers are equipped with $M$ antennas, to achieve a full DoF of $M$ per user, i.e., DoF of $2M$ in total, cooperating node two, must at least receive the backhaul load of  $M$ from the backhaul network, and consequently the cooperating node one should at least send the backhaul load of $M$.
Considering \eqref{eq:MIMO_dof_bound_2} and with the same argument, Claim~\ref{claim:number_messages} is true for the case of two--user interference channels.

Now, consider a two--user interference channel, where transmitter one and its corresponding receiver are equipped with $M$ antennas (i.e., $N_1 = M_1 = M$), and transmitter two and its corresponding receiver are equipped with $(K-1)M$ antennas (i.e., $N_2 = M_2 = (K-1)M$).
Under \nameref{cond2} (Condition~\ref{cond2}) and~\nameref{cond:ehc} (Condition~\ref{cond:ehc}), and according to \eqref{eq:MIMO_dof_bound_1}, in the above interference channel, in order to achieve the total DoF of $KM$, receiver two must at least receive the backhaul load of  $M$ from the backhaul network, and consequently receiver one should at least send the backhaul load of $M$.
Considering \eqref{eq:MIMO_dof_bound_2}, one also illustrates that receiver one must at least receive the backhaul load of $M$.

Finally, consider a $K$--user interference channel, where all the transmitters and receivers are equipped with $M$ antennas.
Let a genie perform the required backhaul communications among all the cooperating nodes except for the communications to or from cooperating node $i$.
The system now is equivalent to the system introduced in the last paragraph, and therefore the cooperating node $i$ must send and receive a backhaul load greater than $M$.
The genie is free to choose the individual cooperating node, therefore each of the cooperating nodes requires to send and receive a backhaul load greater than $M$ through the backhaul network.
The genie aided scenario provides an upper bound for the required backhaul load of the original system and therefore, the proof is complete.

\end{proof}

\begin{claim}\label{claim:central}
In a $K$--user interference channel with backhaul cooperation, in the presence of \nameref{cond2} (Condition~\ref{cond2}) and~\nameref{cond:ehc} (Condition~\ref{cond:ehc}), in order to achieve the total DoF of $M$ and for the central processor to have access to the whole set of messages with vanishing probability of error, the central processor must receive a backhaul load greater than $M(K-1)$. 
\end{claim}
\begin{proof}[Proof of Claim~\ref{claim:central}]
Let cooperating node $i$ play the role of the central processor.
Using Fano's inequality we have
\begin{align*}
H(W_1,W_2,\ldots,W_K|V_i , M_i^{[n]}) &\leq \epsilon_n,\\
H(W_1,W_2,\ldots,W_K) - I(W_1,W_2,\ldots,W_K; V_i, M_i^{[n]}) &\leq \epsilon_n,\\
H(W_1,W_2,\ldots,W_K)  &\leq I(W_1,W_2,\ldots,W_K; V_i, M_i^{[n]}) +  \epsilon_n,\\
H(W_1,W_2,\ldots,W_K)  &\leq H(V_i) + H(M_i^{[n]}) +  \epsilon_n,
\end{align*}
where, in case of receivers cooperation $V_i$ denotes the set of all received signals form the wireless interface at receiver $i$ and in case of transmitters cooperation $V_i = W_i$.
$M_i^{[n]}$ also is defined as in~\eqref{eq:previous_messages}.
In  high $\SNR$ regimes, we have
\begin{align*}
(K-1)M \leq  \lim_{P\rightarrow\infty} \frac{H(M_i^{[n]})}{\log(P)} + \hat{\epsilon}_n.
\end{align*}
Therefor, in order for the central processor to have access to the whole set of messages with vanishing probability of error, it must receive a backhaul load greater than or equal to $M(K-1)$.

\end{proof}

Our goal is to achieve the DoF of $M$ per user and therefore only $o(K)$ of the users are allowed to get the individual DoF less than $M$.
We simply ignore the messages to or from the cooperating nodes corresponding to these $o(K)$ users.
However, we still have $\bar{K} = K-o(k)$ users, each of which requires to achieve the individual DoF of $M$.
Assume that $d_j$ indicates the connectivity degree of node $j$, i.e., $\text{Deg}(j) = d_j$ for all $j$ and let $i = \underset{j}{\arg\max}( d_j)$.

Taking node $j$ as the central processor, for the total backhaul load, $K \alpha$, we have,
\begin{align}\label{eq:Kalpha}
K\alpha \geq M\big(d_j + 2 (\bar{K} - 1 - d_j ) \big) + M(\bar{K}-1).
\end{align}
By Definition~\ref{def:centralized_scheme}, the central processor $j$ requires to decode the whole set of messages.
In order to do so, and according to Claim~\ref{claim:central}, it requires to receive a backhaul load of $M(\bar{K}-1)$.
However, since before any cooperation each individual node has access to at most $M$ DoF, node $j$ is able to receive a backhaul load of $d_jM$, with a single hop communication, while for the rest (i.e., $M(\bar{K} - 1 - d_j)$) at least two hops is required.
On the other hand, according to Claim~\ref{claim:number_messages}, each of the cooperating nodes are required to receive the backhaul load of $M$, which leads to the second term in the right hand side of the above inequality, i.e., $M(\bar{K}-1)$.
We ignore the backhaul messages from and to the receivers achieving individual DoF less than $M$.

Note that the right hand side of \eqref{eq:Kalpha} is decreasing with respect to $d_j$ and therefore, the tightest bound for the normalized backhaul load $\alpha^*$ corresponds to $d_i$, and we have,
\begin{align*}
K \alpha^*  \geq  3M(\bar{K}-1) - Md_i,
\end{align*}
to achieve a DoF of $M$ per user.

Therefore, in order to achieve a DoF of $M$ per user with a backhaul load of $2M$ per user, it is required to have at least one node $i$, with $d_i=K-o(K)$.

\end{proof}

We have the following Corollary for finite values of $K$.
\begin{corollary}\label{col:1}
For the $K$--user interference channel with finite values of $K$, and in the presence of \nameref{cond2} (Condition~\ref{cond2}) and~\nameref{cond:ehc} (Condition~\ref{cond:ehc}), the class of $(2,1)$--centralized schemes is feasible, if and only if in the backhaul graph, there exists at least one node, say $i$, with $K-1$ degrees of connectivity, i.e.,
\begin{align*}
\textnormal{Deg} (i) = K-1.
\end{align*}
\end{corollary}

Note that the set of conditions in Lemma \ref{prop:suf_nec}, is equivalent to the following condition.
\begin{condition}
In the backhaul equivalent graph $G=(\mathcal{V},\mathcal{E})$, there exists a subset of nodes, $\mathcal{V}_c\subseteq \mathcal{V}$ where $|\mathcal{V}_c|= K - o(K)$ and in the induced subgraph $G[\mathcal{V}_c] = (\mathcal{V}_c , \mathcal{E}_c)$, there exists at least one node with Normalized Closeness Centrality measure equal to one.
\end{condition}

In fact, the maximum amount of information that we are able to gather, corresponds to the size of the subsets of nodes which are  connected.
Therefore, in order to achieve a DoF of one per user by exploiting the class of $(2M,M)$--centralized schemes, it is required to have a connected subset of nodes $\mathcal{V}_c$ with $|\mathcal{V}_c|= K-o(K)$.

According to Definition~\ref{def:closeness_centrality}, since there exists a node, say $i$, with \emph{closeness centrality measure} equal to one, one concludes that $\sum_{j}d(i,j) = K- 1 - o(K)$.
We choose node $i$ as the central processor, and according to Definition~\ref{def:centralized_scheme}, it must have access to the whole set of messages.
According to Claim~\ref{claim:central} and the fact that before the cooperation each node has access to at most a DoF of $M$ to gather all information, a backhaul load of $M\sum_{j}d(i,j)$ is required. 
On the other hand, according to Claim~\ref{claim:number_messages}, at least a backhaul load of $M(K-1)$ is required to be received by the other cooperating nodes.
Therefore, the backhaul load $\alpha$ equals $\frac{M(\sum_{j}d(i,j)+K-1)}{K}$ which is asymptotically equal to $2M$.

Also note that, the condition in Corollary~\ref{col:1}, is equivalent to the case that  a ``Star'' sub-graph can be extracted from the equivalent backhaul graph.

\section{The Polynomial Time Algorithm}\label{sec:com}

In this section, we first prove Theorem~\ref{theorem:poly} for interference channels with even number of users, i.e., $K=2m$, and then extend the results to the case of odd number of users.

\begin{lemma}\label{lemma:IS-cond1}
Let $G$ be the equivalent bipartite graph of a $K$--user interference channel.
Then, \nameref{cond:ehc} (Condition~\ref{cond:ehc}) holds if and only if the size of the largest PIS of $G$ is no larger than $\ceil{\frac{K}{2}}$.
\end{lemma}
\begin{proof}
Proof of necessity: Suppose that there exists a PIS, $\mathcal{S}$, with $|\mathcal{S}|=\ceil{\frac{K}{2}}+1$. 
Let $l=\big| \mathcal{S} \cap \mathcal{T} \big|$. 
Then, $l\le m$ and $\big| \mathcal{R} \backslash \mathcal{S} \big| = K - \big| \mathcal{R} \cap \mathcal{S} \big| = \floor{\frac{K}{2}}+l-1$. 
Since $\mathcal{S}$ is a PIS, the neighbouring receivers of $\mathcal{S} \cap \mathcal{T}$ are contained in $\mathcal{R} \backslash \mathcal{S} $. 
Therefore, the set  $\mathcal{S} \cap \mathcal{T}$ of $l$ transmitters has at most $\big| \mathcal{R} \backslash \mathcal{S} \big| = \floor{\frac{K}{2}}+l-1$ neighbours, contradicting  \nameref{cond:ehc} (Condition~\ref{cond:ehc}).

Proof of sufficiency: Assume that \nameref{cond:ehc} (Condition~\ref{cond:ehc}) does not hold, i.e., there exists a subset $\mathcal{S}\subset\mathcal{T}$, with $|\mathcal{S}| = \ell$ ($\ell \leq\ceil{\frac{K}{2}}$), where $|\mathcal{N}_{\mathcal{R}}(\mathcal{S})|\ \textless\ \floor{\frac{K}{2}}+l$, or equivalently $\big|\mathcal{R}\backslash \mathcal{N}_{\mathcal{R}}(\mathcal{S})\big|\ \textgreater\ \ceil{\frac{K}{2}}-\ell$.
Then, $\mathcal{S} \cup \big\{\mathcal{R}\backslash N_{\mathcal{R}}(\mathcal{S}) \big\}$ forms an independent set of size larger than $\ceil{\frac{K}{2}}$.
\end{proof}

\begin{definition}\label{def:G_hat}
Consider a transmitter $a\in \mathcal{T}$ and a receiver $b\in \mathcal{R}\backslash N_\mathcal{R}(a)$, i.e, $L_{ba} = 0$.
We define a graph $\hat{G}(a,b)$, by modifying the equivalent bipartite graph $G$ as follows.
\begin{enumerate}
\item Add a source node $s$ and a sink node $t$.
\item Connect node $s$ to all the transmitters and node $t$ to all the receivers, with one--way links.
\item Set the weights of the links in $G$, the link connecting $s$ to $a$ and the link connecting $b$ to $t$, to infinite (or large enough, e.g., $\ceil{\frac{K}{2}}^3$).
\item  Set the weights of all the links connecting $s$ to $\mathcal{T}\backslash a$ and the links connecting $\mathcal{R}\backslash b$ to $t$ equal to one.
\end{enumerate}
\end{definition}

\begin{lemma}\label{claim:min-cut}
Consider a transmitter $a\in \mathcal{T}$ and a receiver $b\in \mathcal{R}\backslash N_\mathcal{R}(a)$, and the corresponding graph $\hat{G}(a,b)$ as in Definition~\ref{def:G_hat}.
Then,
\begin{align}\label{eq:min-cut}
\big|\min\textnormal{--cut}(\hat{G}(a,b))\big| = 2K - \big|\max\textnormal{--PIS}(a,b)
\big|.
\end{align}
\end{lemma}
\begin{proof}
Let $\mathcal{S}$ be a cut of $\hat{G}(a,b)$ which consists only of unit capacity links. 
Let  $\mathcal{I}$ be the set of nodes that are endpoints of no link in $\mathcal{S}$. 
Then, it follows from the definition of a cut that no two nodes in $\mathcal{I}$ have any links in between, and $\mathcal{I}$ is an independent set of $\hat{G}(a,b)$ containing $a$ and $b$. 
Conversely, if  $\mathcal{I}$ is an independent set of $\hat{G}(a,b)$ containing $a$ and $b$, then removing the set $\mathcal{S}$ of unit capacity links with no endpoint in $\mathcal{I}$ will disconnect $s$ from $t$.
As a result, $\mathcal{S}$ is a cut of $\hat{G}(a,b)$ with $|\mathcal{S}| = 2K - |\mathcal{I}|$.

Therefore, each cut $\mathcal{S}$ of $\hat{G}(a,b)$ corresponds to an independent set $\mathcal{I}$ of size $2K-|\mathcal{S}|$ containing $a$ and $b$, and vice-versa.
Thus, $\min$--cut$(\hat{G}(a,b))$ corresponds to $\max$--PIS$(a,b)$, and the lemma follows. 
\end{proof}

\begin{lemma}\label{lemma:poly-pis}
For a transmitter $a$ and a receiver $b$,
$\max$--PIS$(a,b)$ can be found in a time polynomial in the network size, $K$.
\end{lemma}

\begin{proof}
If  $a$ and  $b$ are connected, then there is no independent set including both of them.
Let us now focus on the case that $L_{ba}=0$.
According to Lemma~\ref{claim:min-cut}, in order to find $\max$--PIS$(a,b)$, we only need to find the $\min$--cut$(\hat{G}(a,b))$.
On the grounds of \emph{Max-Flow Min-Cut} theorem \cite{dantzig2003max}, this is equivalent to finding the maximum flow from $s$ to $t$.
The later can be solved through linear programming and thereby is polynomial time.
\end{proof}

\begin{proof}[Proof of Theorem~\ref{theorem:poly}]
According to Lemma~\ref{lemma:IS-cond1}, it is sufficient to find the size of maximum PIS for the equivalent bipartite graph, $G$, to verify \nameref{cond:ehc} (Condition~\ref{cond:ehc}).
However, according to Lemma~\ref{lemma:poly-pis}, it takes polynomial time to find the $\max$--PIS$(a,b)$, for each arbitrary transmitter $a$ and receiver $b$.
Solving this polynomial time problem for $K^2$ possible ways of choosing transmitter and receiver pairs, we can find the maximum PIS in a time polynomial in $K$.
This completes the proof.
\end{proof}

\appendices
\section{Preliminaries on Linear Algebra}\label{app:linear_algebra}

In this Appendix, we provide some preliminaries required to follow the proofs in this paper.
These results are purely mathematical and hence, are provided in a separate Appendix. 

\begin{lemma}\label{lemma:block_det}
Assume that $\vm{A}$ and $\vm{D}$ are diagonal matrices with sizes $r\times r$ and $n\times n$ respectively.
If
\begin{align*}
\vm{M} = \left[ \begin{array}{c c}
\vm{A} & \vm{B}\\
\vm{C} & \vm{D}
\end{array}\right],
\end{align*}
then
\begin{align}
\det \vm{M} = \det\vm{A} \det\vm{D},
\end{align}
whenever at least one of the blocks $\vm{B}$ and $\vm{C}$ is equal to $\vm{0}$.
\end{lemma}

The proof is base on induction.
We first consider the case $n=1$, where the result becomes obvious.
Then assuming that the result holds for $n=m-1$, we show that the result holds for the case of $n=m$, by expanding the determinant.
We do not go into details of the proof. 

\begin{lemma}\label{lemma:det_block}
For block matrix
$\vm{M} = \left[ \begin{array}{c c}
\vm{A} & \vm{B}\\
\vm{C} & \vm{D}
\end{array}\right]$ 
with submatrices $\vm{A}_{n\times n}$, $\vm{B}_{n\times m}$, $\vm{C}_{m\times n}$ and $\vm{D}_{m\times m}$, we have:
\begin{align}
\det(\vm{M}) = \left\lbrace \begin{array}{c c}
\det(\vm{A})\det(\vm{D}-\vm{C}\vm{A}^{-1}\vm{B}), & \text{if}\ \vm{A}\ \text{is invertible},\\
\det(\vm{D})\det(\vm{A}-\vm{B}\vm{D}^{-1}\vm{C}), & \text{if}\ \vm{D}\ \text{is invertible}.
\end{array}\right.
\end{align}
\end{lemma}

\begin{proof}
Recall the well-known fact that $det(\vm{X}\vm{Y}) = \det(\vm{X})\det(\vm{Y})$.\\
Consider the case that $\vm{A}$ is invertible.
The proof is based on the fact that 
\begin{align*}
\left[ \begin{array}{c c}
\vm{A} & \vm{B}\\
\vm{C} & \vm{D}
\end{array}\right]
\left[ \begin{array}{c c}
\vm{I}_n & -\vm{A}^{-1}\vm{B}\\
\vm{0} & \vm{I}_m
\end{array}\right] = 
\left[ \begin{array}{c c}
\vm{A} & \vm{0}\\
\vm{C} & \vm{D}-\vm{C}\vm{A}^{-1}\vm{B}
\end{array}\right].
\end{align*}
From Lemma~\ref{lemma:block_det}, we have 
\begin{align*}
\det\left[ \begin{array}{c c}
\vm{I}_n & -\vm{A}^{-1}\vm{B}\\
\vm{0} & \vm{I}_m
\end{array}\right] = 1,
\end{align*}
and
\begin{align*}
\det\left[ \begin{array}{c c}
\vm{A} & \vm{0}\\
\vm{C} & \vm{D}-\vm{C}\vm{A}^{-1}\vm{B}
\end{array}\right] = \det(\vm{A})\det(\vm{D}-\vm{C}\vm{A}^{-1}\vm{B}),
\end{align*}
which completes the proof.
Considering the case that $\vm{D}$ is invertible, the proof is based on the fact that 
\begin{align*}
\left[ \begin{array}{c c}
\vm{A} & \vm{B}\\
\vm{C} & \vm{D}
\end{array}\right]
\left[ \begin{array}{c c}
\vm{I} & \vm{0}\\
-\vm{D}^{-1}\vm{C} & \vm{I}
\end{array}\right] = 
\left[ \begin{array}{c c}
\vm{A} - \vm{B}\vm{D}^{-1}\vm{C} & \vm{B}\\
\vm{0} & \vm{D}
\end{array}\right].
\end{align*}
The rest of the proof is similar to the case where $\vm{A}$ is invertible.
\end{proof}

\begin{lemma}[\cite{ashraphijuo2013capacity}, Lemma 11\ ]\label{lemma:matrix_function}
Let $\vm{L}(\vm{K},\vm{S})$ be defined as
\begin{align}
\vm{L}(\vm{K},\vm{S}) \triangleq \vm{K} - \vm{K}\vm{S}\big(\vm{I}_N+ \vm{S}^\dagger \vm{K}\vm{S}\big)^{-1}\vm{S}^\dagger\vm{K},
\end{align}
for some $M\times M$ p.s.d. Hermitian matrix $\vm{K}$ and some $M\times N$ matrix $\vm{S}$.
Then, if $0\preceq \vm{K}_1 \preceq\vm{K}_2$ for some Hermitian matrices $\vm{K}_1$ and $\vm{K}_2$, we have
\begin{align}
\vm{L}(\vm{K}_1,\vm{S}) \preceq \vm{L}(\vm{K}_2,\vm{S}).
\end{align}
\end{lemma}

\begin{lemma}[\cite{horn1990matrix}, Corollary 7.7.4.b and \cite{ashraphijuo2014capacity}, Lemma 2\ ]\label{lemma:A_B}
For two Hermitian positive definite matrices $\vm{A}$ and $\vm{B}$ of size $m\times m$, if $\vm{A}\preceq \vm{B}$, then $\det(\vm{A}) \leq \det(\vm{B})$.
\end{lemma}

\begin{lemma}\label{lemma:psd_trace}
For Hermitian positive semi--definite matrix $\vm{Q}$ with $Tr(\vm{Q}) = \lambda \geq 0$ we have
\begin{align}
\lambda \vm{I} - \vm{Q} \succeq \vm{0}.
\end{align}
\end{lemma}

\begin{proof}
Let the  $\mathcal{\Lambda}= \{\lambda_1 ,\lambda_2, \ldots,\lambda_n \}$ denote the set the eigenvalues of $\vm{Q}$. 
Since $\vm{Q}$ is a positive semi--definite matrix, we have $\lambda_i\geq 0$ for all $i$. 
On the other hand, we know that $Tr{\vm{Q}} = \sum_{i} \lambda_i$, and therefore one can conclude that $\lambda_i\leq \lambda$ for all $i$.
As the result, the all the eigenvalues of the matrix $\lambda \vm{I} - \vm{Q}$ are greater than zero, i.e., $\lambda \vm{I} - \vm{Q}$ is positive semi--definite.
\end{proof}

\begin{lemma}[\cite{karmakar2012generalized}, Lemma 4\ ]\label{lemma:dof_simple_sum}
Let $\vm{H}_{1}\in \mathbb{C}^{N\times M_1}$,$\vm{H}_{2}\in \mathbb{C}^{N\times M_2}$,\ldots , and $\vm{H}_{k}\in \mathbb{C}^{N\times M_k}$ be $k$ full--rank and independent channel matrices.
Then, the following holds
\begin{align}\label{eq:dof_simple_sum}
\log&\det(\vm{I}_N + P \vm{H}_1\vm{H}_1^\dagger + P \vm{H}_2\vm{H}_2^\dagger + \ldots+ P \vm{H}_k\vm{H}_k^\dagger) \nonumber \\
&= \log\det(\vm{I}_N + P[\vm{H}_1\ \ldots\ \vm{H}_k][\vm{H}_1\ \ldots\ \vm{H}_k]^\dagger)\nonumber \\
&= \min\{ N, M_1 + M_2 + \dots + M_k\}\log P + o(\log P)
\end{align}
\end{lemma}

\begin{lemma}[\cite{ashraphijuo2013capacity}, Lemma 23\ ]\label{lemma:dof_complex}
Let $\vm{H}_{ii}\in\mathbb{C}^{N_i\times M_i}$ and $\vm{H}_{ji}\in\mathbb{C}^{N_j\times M_i}$ be two channel matrices whose entries are independently chosen from $\mathsf{CN}(0,1)$.
Then, the following holds with probability 1 (over the randomness of channel matrices).
\begin{align}\label{eq:dof_complex}
\log\det(\vm{I}_{N_i} &+ P \vm{H}_{ii}\vm{H}_{ii}^\dagger - P^2 \vm{H}_{ii}\vm{H}_{ji}^\dagger (\vm{I}_{N_j}+P\vm{H}_{ji}\vm{H}_{ji}^\dagger)^{-1}\vm{H}_{ji}\vm{H}_{ii}^\dagger)\nonumber\\
&= \min\{N_i,(M_i-N_j)^+\}\log P + o(\log P)
\end{align} 
\end{lemma}

\section{Proof of Lemma~\ref{lemma:MIMO_dof_bound}}\label{app:00}
In this Appendix, we first provide some preliminary information theoretic results.
Subsequently, some upper bounds for the sum achievable rate of two--user multiple antenna interference channel with backhaul cooperation, is presented in Lemma~\ref{lemma:entropy_bound} and Lemma~\ref{lemma:bound_log_det}.
Finally, we provide a proof for Lemma~\ref{lemma:MIMO_dof_bound}.
Note that Lemma~\ref{lemma:MIMO_dof_bound}, works both for receivers cooperation and transmitters cooperation. 

\begin{definition}
For a two--user interference channel, we define $\vm{s}_i(t)$ as the random variable corresponding to the undesired part of the received signal by receiver $j$ at time $t$, i.e., 
\begin{align}\label{eq:s}
\vm{s}_i(t) \triangleq \vm{H}_{ji} \vm{x}_i(t) + \vm{z}_j(t),
\end{align}
for $i \in \{1,2\}$ and $j\neq i$.
\end{definition}

\begin{definition}
We define $\vm{Q}_{ij}(t)$ as the cross correlation between the signals sent from transmitter $i$ and transmitter $j$ at time $t$, i.e.,
\begin{align*}
\vm{Q}_{ij}(t) \triangleq \mathbb{E}\{ \vm{x}_i(t) \vm{x}_j(t)^\dagger\}.
\end{align*}
We also define $\bar{\vm{Q}}_{ij}$ as the summation of $\vm{Q}_{ij}(t)$ over $t$, i.e.,
\begin{align}\label{eq:bar_Q_ij}
\bar{\vm{Q}}_{ij} \triangleq \frac{1}{n} \sum_{t=1}^n \vm{Q}_{ij}(t)
\end{align}
Note that, for the case of receivers cooperation, $\vm{Q}_{ij}(t) = 0$ for all $i\neq j$, since transmitters cannot cooperate.
\end{definition}

\begin{definition}
We define $\vm{Q}_t$ as
\begin{align}\label{eq:def_Q}
\vm{Q}(t)  = \left[\begin{array}{cc}
								\vm{Q}_{11}(t) & \vm{Q}_{12}(t)\\
								\vm{Q}_{21}(t) & \vm{Q}_{22}(t)
							\end{array} \right],
\end{align}
and $\bar{\vm{Q}}$ as the summation of $\vm{Q}(t)$ over $t$, i.e.,
\begin{align}\label{eq:bar_Q}
\bar{\vm{Q}} \triangleq \frac{1}{n} \sum_{t=1}^n \vm{Q}(t) = \left[\begin{array}{cc}
								\bar{\vm{Q}}_{11} & \bar{\vm{Q}}_{12}\\
								\bar{\vm{Q}}_{21} & \bar{\vm{Q}}_{22}
							\end{array} \right].
\end{align} 
\end{definition}

\begin{lemma}{(\cite{shang2010capacity}, Lemma 2 and \cite{ashraphijuo2014capacity}, Lemma 8)} \label{lemma:guassian}
Let $\vm{x}$ and $\vm{y}$ be two random vectors, and $\vm{x}^G$ and $\vm{y}^G$ be Gaussian vectors with covariance matrices satisfying
\begin{align}
Cov\left[\begin{array}{c}
\vm{x}\\
\vm{y}
\end{array}\right] = Cov \left[\begin{array}{c}
\vm{x}^G\\
\vm{y}^G
\end{array}\right],
\end{align}

Then, we have
\begin{align}
h(\ \vm{y}\ ) &\leq h(\ \vm{y}^G\ ),\\
h(\ \vm{y}\ |\ \vm{x}\ ) &\leq h(\ \vm{y}^G\ |\ \vm{x}^G\ ).
\end{align}
\end{lemma}

\begin{definition}
We define the vectors $\vm{x}_1^G$ and $\vm{x}_2^G$ as Gaussian random vectors with the covariance matrix
\begin{align}
Cov\left[\begin{array}{c}
\vm{x}_1^G\\
\vm{x}_2^G
\end{array}\right] = Cov \left[\begin{array}{c}
\vm{x}_1\\
\vm{x}_2
\end{array}\right],
\end{align}
and correspondingly we define $\vm{s}_i^G$ by replacing $\vm{x}_i$ with $\vm{x}_i^G$ in \eqref{eq:s}.
\end{definition}

\begin{lemma}\label{lemma:concavity}
Let $\vm{x}^G$, $\vm{z}_1$ and $\vm{z}_2$ be independent Gaussian random vectors, where $\vm{Q}$ is the covariance matrix of $\vm{x}^G$.
Then, the conditional entropy
$
h(\ \vm{H}_1\ \vm{x}^G +\ \vm{z}_1\ |\ \vm{H}_2\ \vm{x}^G +\ \vm{z}_2\ ),
$
is concave on $\vm{Q}$, for some deterministic matrices $\vm{H}_1$ and $\vm{H}_2$.
\end{lemma}

\begin{proof}

Let $T$ be a Bernoulli random variable with parameter $p$ and $\vm{Q} = (1-p)\ \vm{P}_0 + p\ \vm{P}_1$.
We define the function $f(\vm{Q}) = h(\ \vm{H}_1\ \vm{x}^G + \vm{z}_1\ |\ \vm{H}_2\ \vm{x}^G + \vm{z}_2\ )$.
Assume that conditioning on $T$ we have
\begin{align*}
\vm{Q} = \left\lbrace \begin{array}{c c}
\vm{P}_0, & \text{if}\ T=0,\\
\vm{P}_1, & \text{if}\ T=1,
\end{array}\right.
\end{align*}
and therefore,
\begin{align*}
h(\vm{H}_1\ \vm{x}^G &+ \vm{z}_1\ | \vm{H}_2\ \vm{x}^G +\vm{z}_2 , T\ )\\
&=(1-p)\ h(\vm{H}_1\ \vm{x}^G + \vm{z}_1\ | \vm{H}_2\ \vm{x}^G +\vm{z}_2 , T=0 ) + 
p\ h(\vm{H}_1\ \vm{x}^G + \vm{z}_1\ | \vm{H}_2\ \vm{x}^G +\vm{z}_2 , T=1 )\\
&= (1-p)\ f(\vm{P}_0) + p\ f(\vm{P}_1).
\end{align*}
On the other hand, since conditioning reduces the entropy, we have,
\begin{align*}
h(\vm{H}_1\ \vm{x}^G + \vm{z}_1 | \vm{H}_2\ \vm{x}^G +\vm{z}_2 , T ) \leq h(\vm{H}_1\ \vm{x}^G + \vm{z}_1 | \vm{H}_2\ \vm{x}^G +\vm{z}_2),
\end{align*}
or
\begin{align*}
(1-p)\ f(\vm{P}_1) + p\ f(\vm{P}_2) \leq f\big((1-p)\ \vm{P}_0\ + p\ \vm{P}_1\big)
\end{align*}
which is the definition of a concave function and the proof is complete.

\end{proof}

As a direct consequence of Lemma~\ref{lemma:concavity}, by setting $\vm{H}_1 = \vm{I}$ and $\vm{H}_2 = 0$ we have the following corollary.
\begin{corollary}\label{corollary:concavity}
Let $\vm{x}^G$ and $\vm{z}$ be independent Gaussian random vectors, where $\vm{Q}$ is the covariance matrix of $\vm{x}^G$.
Then, the entropy $h(\ \vm{x}^G\ +\ \vm{z}\ )$ is concave on $\vm{Q}$.
\end{corollary}

In the following lemma, we give the first set of bounds on the sum achievable rate of a two--user interference channel with backhaul cooperation, with respect to the entropy of each sample of the encoded message.

\begin{lemma} \label{lemma:entropy_bound}
For a two user interference channel with backhaul cooperation, we have,
\begin{align}
R_1 + R_2 \leq \frac{1}{n}  \sum_{t=1}^n \bigg(\ h \big(\ \vm{H}_{11} \vm{x}_1(t) + \vm{z}_1(t)\ |\ \vm{s}_1(t)\ \big) +h \big(\ \vm{y}_2(t)\ \big) - h\big(\ \vm{z}_1(t) , \vm{z}_2(t)\ \big) + R_B^{[1,2]}\ \bigg) \label{eq:ent_bound_1}  \\
R_1 + R_2 \leq \frac{1}{n}  \sum_{t=1}^n \bigg(\ h \big(\ \vm{H}_{22} \vm{x}_2(t) + \vm{z}_2(t)\ |\ \vm{s}_2(t)\ \big) +h \big(\ \vm{y}_1(t)\ \big) - h\big(\ \vm{z}_1(t) , \vm{z}_2 (t)\ \big) + R_B^{[2,1]}\ \bigg) \label{eq:ent_bound_2}
\end{align}
\end{lemma}

\begin{proof} 
We prove Lemma~\ref{lemma:entropy_bound} separately for receiver and transmitters cooperation scenarios.

\textbf{Receivers cooperation.}
Assume that a genie gives information $\vm{x}_2^n$ and $\vm{y}_2^n$ to receiver one.
Now we make use of Fano and processing inequalities in addition to the fact that the backhaul messages from receiver one to receiver two, are solely a function of $\vm{y}_1^n$ and $\vm{y}_2^n$ to complete the proof.
Recall that, according to \eqref{eq:previous_messages}, for a two--user interference channel with backhaul cooperation $M_{i}^{[n]}$ indicates the set of all backhaul messages from receiver $j$ to receiver $i$.

For the rate pair $(R_1,R_2)$ to be achievable we have,
\begin{align*}
n(R_1 + R_2 - \epsilon_n) &\leq I(\vm{x}_1^n ; \vm{y}_1^n , M_{1}^{[n]}) + I(\vm{x}_2^n ; \vm{y}_2^n , M_2^{[n]})\\
&\overset{(a)}{\leq} I(\vm{x}_1^n ; \vm{y}_1^n , M_1^{[n]} , \vm{y}_2^n , \vm{x}_2^n) + I(\vm{x}_2^n ; \vm{y}_2^n) +I(\vm{x}_2^n;M_2^{[n]} | \vm{y}_{2}^n)\\
&\overset{(b)}{\leq} I(\vm{x}_1^n ; \vm{y}_1^n , M_1^{[n]} , \vm{y}_2^n | \vm{x}_2^n) + h(\vm{y}_2^n) - h(\vm{s}_1^n) + H(M_2^{[n]})\\
&\overset{(c)}{=} I(\vm{x}_1^n ; \vm{y}_1^n  , \vm{y}_2^n | \vm{x}_2^n) + h(\vm{y}_2^n) - h(\vm{s}_1^n) + H(M_2^{[n]})\\
&= h(\vm{H}_{11}\vm{x}_1^n + \vm{z}_1^n,\vm{s}_1^n)-h(\vm{z}_1^n, \vm{z}_2^n)+h(\vm{y}_2^n)-h(\vm{s}_1^n)+H(M_2^{[n]})\\
&= h(\vm{H}_{11}\vm{x}_1^n + \vm{z}_1^n|\vm{s}_1^n)-h(\vm{z}_1^n , \vm{z}_2^n)+h(\vm{y}_2^n)+H(M_2^{[n]})\\
&\overset{(d)}{\leq} \sum_{t=1}^n \bigg(\ h\big(\ \vm{H}_{11}\vm{x}_1(t) + \vm{z}_1(t)\ |\ \vm{s}_1^n\ \big) +h \big(\ \vm{y}_2(t)\ \big) - h\big(\ \vm{z}_1(t) , \vm{z}_2(t)\ \big) \bigg) + H \big(\ M_2^{[n]}\ ) \\
&\overset{(e)}{=} \sum_{t=1}^n \bigg(\ h\big(\ \vm{H}_{11}\vm{x}_1(t) + \vm{z}_1(t)\ |\ \vm{s}_1(t)\ \big) +h \big(\ \vm{y}_2(t)\ \big) - h\big(\ \vm{z}_1(t) , \vm{z}_2(t)\ \big) + R_B^{[1,2]} \bigg),
\end{align*}
where $\epsilon_n\rightarrow 0$ as $n\rightarrow \infty$.
(a) is due to the genie giving side information $(\vm{x}_2^n,\vm{y}_2^n)$ to receiver one and the chain rule.
(b) is due to the chain rule, the independence of $\vm{x}_1^n$ and $\vm{x}_2^n$, and the fact that $I(x,y|z)\leq H(y)$.
(c) is also due to the chain rule and the fact that $M_1^{[n]}$ is solely a function of $(\vm{y}_1^n,\vm{y}_2^n)$.
(d) is due to the fact that the independence maximizes the mutual information and that the noise is i.i.d., and finally (e) is due to the fact that the channel is memoryless and the definition of $R_B^{[1,2]}$ at \eqref{eq:backhaul_rate}.
Hence and similarly (by giving side information $\vm{x}_1^n$ and $\vm{y}_1^n$ to receiver two) bounds \eqref{eq:ent_bound_1} and \eqref{eq:ent_bound_2} are shown.

\textbf{Transmitters cooperation.} Assume that a genie gives information $\vm{s}_1^n$, $M_2^{[n]}$ and $W_2$ to receiver one, and $M_2^{[n]}$ is available at transmitter two.
We again make use of Fano and processing inequalities in addition to the fact that the backhaul messages $M_2^{[n]}$ is solely a function of $W_1$ and $W_2$ to complete the proof.

For the rate pair $(R_1,R_2)$ to be achievable we have,
\begin{align*}
n(R_1 + R_2 - \epsilon_n) &\leq I(W_1 ; \vm{y}_1^n) + I(W_2 ; \vm{y}_2^n)\\
&\overset{(a)}{\leq} I ( W_1 ; \vm{y}_1^n,\vm{s}_1^n,M_2^{[n]},W_2) + I (W_2 , M_2^{[n]} ; \vm{y}_2 ^ n) \\
&\overset{(b)}{=} I(W_1;\vm{y}_1^n,\vm{s}_1^n,M_2^{[n]} | W_2) + I (W_2 , M_2^{[n]} ; \vm{y}_2^n) \\
&= I (W_1; M_2^{[n]} | W_2) + I(W_1 ; \vm{y}_1^n , \vm{s}_1^n | M_2^{[n]} ,W_2)+ h(\vm{y}_2^n) - h(\vm{y}_2^n| W_2,M_2^{[n]})\\
&=H(M_2^{[n]}|W_2)-H(M_2^{[n]}|W_1,W_2) + h(\vm{y}_1^n,\vm{s}_1^n|M_2^{[n]},W_2)-h(\vm{y}_1^n,\vm{s}_1^n|M_2^{[n]},W_2,W_1)\\
&\quad\quad + h(\vm{y}_2^n)- h(\vm{s}_1^n|W_2,M_2^{[n]})\\
&\overset{(c)}{=} H(M_2^{[n]}|W_2)+h(\vm{y}_1^n,\vm{s}_1^n|M_2^{[n]},W_2)-h(\vm{z}_1^n,\vm{z}_2^n)+h(\vm{y}_2^n)-h(\vm{s}_1^n|M_2^{[n]}, W_2)\\
 &= h(\vm{y}_1^n|\vm{s}_1^n,M_2^{[n]},W_2) - h(\vm{z}_1^n , \vm{z}_2^n)+ h(\vm{y}_2^n)+H(M_2^{[n]}|W_2)\\
&\leq h(\vm{H}_{11}\vm{x}_1^n+\vm{z}_1^n |\vm{s}_1^n)-h(\vm{z}_1^n,\vm{z}_2^n)+h(\vm{y}_2^n)  + H(M_2^{[n]})
\end{align*}
where $\epsilon_n\rightarrow 0$ as $n\rightarrow \infty$.
(a) is due to the genie giving side information $(\vm{s}_1^n,M_2^{[n]},W_2)$ to receiver one and the availability of $M_2^{[n]}$ at transmitter two.
(b) is due to the independence of the messages $W_1$ and $W_2$.
(c) is also due to the fact that $M_2^{[n]}$ is a function of $(W_1,W_2)$ only.
The rest of the proof, follows the same line as the proof for the case of receivers cooperation.
Hence, the bounds \eqref{eq:ent_bound_1} and \eqref{eq:ent_bound_2} are obtained.

\end{proof}

Using the power constraint in \eqref{eq:ave_pow_const}, we have the following lemma.
\begin{lemma}
For the matrices $\bar{\vm{Q}}_{ii}$ and $\bar{\vm{Q}}$ defined at \eqref{eq:bar_Q_ij} and \eqref{eq:bar_Q}, respectively, we have
\begin{align}\label{eq:Q_11}
P \vm{I} - \bar{\vm{Q}}_{ii} \succeq \vm{0},
\end{align}
and
\begin{align}\label{eq:Q}
2P \vm{I} - \bar{\vm{Q}} \succeq \vm{0},
\end{align}
\end{lemma}

\begin{proof}
Focusing on $\bar{\vm{Q}}_{ij}$ we have
\begin{align}\label{eq:Q_total}
\begin{split}
Tr(\bar{\vm{Q}}_{ij})= Tr( \frac{1}{n} \sum_{t=1}^n \vm{Q}_{ij}(t)) = \mathbb{E}\{\frac{1}{n} \sum_{t=1}^n Tr(\vm{x}_i(t) \vm{x}_j^\dagger(t)) \}.
\end{split}
\end{align}
Combining \eqref{eq:ave_pow_const} and \eqref{eq:Q_total} results in $Tr(\bar{\vm{Q}}_{11})\leq P$ and $Tr(\bar{\vm{Q}}) = Tr(\bar{\vm{Q}}_{11}) + Tr(\bar{\vm{Q}}_{22})\leq 2P$,
which in addition to the result of Lemma~\ref{lemma:psd_trace} complete the proof.

\end{proof}

In the following lemma, we improve the bounds in \eqref{eq:ent_bound_1} and \eqref{eq:ent_bound_2} and rewrite them in algebraic form.
These new bounds are with respect to the power constraint and channel realizations.

\begin{lemma}\label{lemma:bound_log_det}
For a two--user multiple antenna interference channel with backhaul cooperation, we have
\begin{align}
R_1+R_2 &\leq \log\det \bigg(\ \vm{I}_{N_1} + P\ \vm{H}_{11}\vm{H}_{11}^\dagger - P^2\ \vm{H}_{11}\vm{H}_{21}^\dagger \big(\ \vm{I}_{N_2} + P\ \vm{H}_{21}\vm{H}_{21}^\dagger\ \big)^{-1} \vm{H}_{21}\vm{H}_{11}^\dagger\ \bigg)\nonumber\\
		&\quad + \log\det \big(\ \vm{I}_{N_2} + 2P\ \vm{H}_{22}\vm{H}_{22}^\dagger + 2P\ \vm{H}_{21}\vm{H}_{21}^\dagger\ \big) + R_B^{[1,2]},\label{eq:rate_final_1} \\ 
R_1+R_2 &\leq \log\det \bigg(\ \vm{I}_{N_2} + P\ \vm{H}_{22}\vm{H}_{22}^\dagger - P^2\ \vm{H}_{22}\vm{H}_{12}^\dagger \big(\ \vm{I}_{N_1} + P\ \vm{H}_{12}\vm{H}_{12}^\dagger\ \big)^{-1} \vm{H}_{12}\vm{H}_{22}^\dagger\ \bigg)\nonumber\\
		&\quad + \log\det \big(\ \vm{I}_{N_1} + 2P\ \vm{H}_{11}\vm{H}_{11}^\dagger + 2P\ \vm{H}_{12}\vm{H}_{12}^\dagger\ \big) + R_B^{[2,1]}.\label{eq:rate_final_2}			    
\end{align}
\end{lemma}

\begin{proof}
For any given $t$ we have,
\begin{align}\label{eq:rate_for_t}
\begin{split}
h \big(\ \vm{H}_{11}\vm{x}_1^G(t) &+ \vm{z}_1\ |\ \vm{H}_{21}\vm{x}_1^G(t) + \vm{z}_2\ \big)+ h \big(\ \vm{H}_{21} \vm{x}_1^G(t) + \vm{H}_{22} \vm{x}_2^G(t) + \vm{z}_2\ \big) - h(\ \vm{z}_1,\vm{z}_2\ )\\
&= h\big(\ \vm{H}_{11}\vm{x}_1^G(t)+ \vm{z}_1\ |\ \vm{H}_{21} \vm{x}_1^G(t) + \vm{z}_2\ \big) - h(\ \vm{z}_1\ ) + h\big( \vm{H}_{21} \vm{x}_1^G(t) + \vm{H}_{22} \vm{x}_2^G(t) + \vm{z}_2\ ) - h(\ \vm{z}_2\ )
\end{split}
\end{align}

\begin{claim}\label{claim:ent_1}
For any given $t$
\begin{align}\label{eq:ent_1}
\begin{split}
h\big(\ \vm{H}_{11}\vm{x}_1^G(t)+ \vm{z}_1\ &|\ \vm{H}_{21} \vm{x}_1^G(t) + \vm{z}_2\ \big) - h(\ \vm{z}_1\ )\\
&= \log\det (\vm{I}_{N_1} + \vm{H}_{11} \vm{Q}_{11}(t) \vm{H}_{11}^\dagger - \vm{H}_{11}\vm{Q}_{11}(t) \vm{H}_{21}^\dagger (\vm{I}_{N_2} + \vm{H}_{21}\vm{Q}_{11}(t)\vm{H}_{21}^\dagger)^{-1}\vm{H}_{21}\vm{Q}_{11}(t)\vm{H}_{11}^\dagger).
\end{split}
\end{align}
\end{claim}

\begin{proof}[Proof of Claim~\ref{claim:ent_1}]
In order to prove this claim, we first expand the left hand side of \eqref{eq:ent_1} as
\begin{align}\label{eq:ent_2}
\begin{split}
h\big(\ \vm{H}_{11}\vm{x}_1^G(t)+ \vm{z}_1\ &|\ \vm{H}_{21} \vm{x}_1^G(t) + \vm{z}_2\ \big) - h(\ \vm{z}_1\ )\\
& =  h\big(\ \vm{H}_{11}\vm{x}_1^G(t)+ \vm{z}_1\ ,\ \vm{H}_{21} \vm{x}_1^G(t) + \vm{z}_2\ \big) - h\big(\ \vm{H}_{21} \vm{x}_1^G(t) + \vm{z}_2\ \big) - h(\ \vm{z}_1\ )\\
&= \log\det \left[\begin{array}{ccc}
	\vm{I}_{N_1} + \vm{H}_{11} \vm{Q}_{11}(t) \vm{H}_{11}^\dagger &\hspace{.2in} & \vm{H}_{11} \vm{Q}_{11}(t) \vm{H}_{21}^\dagger\\
	\vm{H}_{21} \vm{Q}_{11}(t) \vm{H}_{11}^\dagger &\hspace{.2in} &\vm{I}_{N_2} + \vm{H}_{21}\vm{Q}_{11}(t)\vm{H}_{21}^\dagger
				  \end{array} \right]\\
&\quad \quad - \log\det (\vm{I}_{N_2} + \vm{H}_{21}\vm{Q}_{11}(t)\vm{H}_{21}^\dagger),
\end{split}
\end{align}
where, using the results from Lemma~\ref{lemma:det_block}, the claim is proved.

\end{proof}

\begin{claim}\label{claim:ent_3}
For any given $t$
\begin{align}\label{eq:ent_3}
\begin{split}
h\big( \vm{H}_{21} \vm{x}_1^G(t) &+ \vm{H}_{22} \vm{x}_2^G(t) + \vm{z}_2\ ) - h(\ \vm{z}_2\ ) =  \log\det \big(\ \vm{I}_{N_2}\ + \vm{H}_2\ \vm{Q}(t)\ \vm{H}_2^\dagger\ \big),
\end{split}
\end{align}
where, $\vm{Q}$ is defined at \eqref{eq:def_Q} and
\begin{align}\label{eq:def_H}
\vm{H}_2  = \left[\begin{array}{cc}
			\vm{H}_{21} & \vm{H}_{22}
			      \end{array} \right].
\end{align} 
\end{claim}

\begin{proof}[Proof of Claim~\ref{claim:ent_3}]
In order to prove this claim, we have
\begin{align}\label{eq:ent_4}
\begin{split}
h\big( \vm{H}_{21} \vm{x}_1^G(t) &+ \vm{H}_{22} \vm{x}_2^G(t) + \vm{z}_2\ \big) - h(\ \vm{z}_2\ )\\ 
&= \log \det \big(\ \vm{I}_{N_2} + \vm{H}_{21}\vm{Q}_{11}(t)\vm{H}_{21}^\dagger + \vm{H}_{22}\vm{Q}_{22}(t) \vm{H}_{22}^\dagger + \vm{H}_{21}\vm{Q}_{12}(t)\vm{H}_{21}^\dagger + \vm{H}_{21}\vm{Q}_{21}(t)\vm{H}_{22}^\dagger\ \big)\\
&= \log\det \big(\ \vm{I}_{N_2} + \left[ \begin{array}{cc}
								\vm{H}_{21} & \vm{H}_{22}
							\end{array}\right]
							\left[\begin{array}{cc}
								\vm{Q}_{11}(t) & \vm{Q}_{12}(t)\\
								\vm{Q}_{21}(t) & \vm{Q}_{22}(t)
							\end{array} \right]
							 \left[ \begin{array}{cc}
								\vm{H}_{21} & \vm{H}_{22}
							\end{array}\right]^\dagger\ \big)\\
&= \log\det \big(\ \vm{I}_{N_2} + \vm{H}_2 \vm{Q}(t) \vm{H}_2^\dagger\ \big),
\end{split}
\end{align}
which, completes the proof.

\end{proof}

\begin{claim}\label{claim:concavity}
Due to the results from Lemma~\ref{lemma:concavity} and Corollary~\ref{corollary:concavity}, the right hand sides of \eqref{eq:ent_1} and \eqref{eq:ent_3} are concave with respect to $\vm{Q}_{11}(t)$ and $\vm{Q}(t)$, respectively.
\end{claim}

Considering \eqref{eq:ent_1}, \eqref{eq:ent_3} and \eqref{eq:rate_for_t}, we have
\begin{align}\label{eq:each_concave}
\begin{split}
\frac{1}{n} \sum_{t=1}^n &\bigg( h \big(\ \vm{H}_{11}\vm{x}_1^G(t) + \vm{z}_1\ |\ \vm{H}_{21}\vm{x}_1^G(t) + \vm{z}_2\ \big)+ h \big(\ \vm{H}_{21} \vm{x}_1^G(t) + \vm{H}_{22} \vm{x}_2^G(t) + \vm{z}_2\ \big) - h(\ \vm{z}_1,\vm{z}_2\ ) \bigg) \\
&= \frac{1}{n}\sum_{t=1}^n \bigg( \log\det (\vm{I}_{N_1} + \vm{H}_{11} \vm{Q}_{11}(t) \vm{H}_{11}^\dagger - \vm{H}_{11}\vm{Q}_{11}(t) \vm{H}_{21}^\dagger (\vm{I}_{N_2} + \vm{H}_{21}\vm{Q}_{11}(t)\vm{H}_{21}^\dagger)^{-1}\vm{H}_{21}\vm{Q}_{11}(t)\vm{H}_{11}^\dagger)\\
&\quad\quad\quad + \log\det \big(\ \vm{I}_{N_2} + \vm{H}_2 \vm{Q}(t) \vm{H}_2^\dagger\ \big) + R_B^{[1,2]} \bigg) \\
& \overset{(a)}{\leq} \log\det (\vm{I}_{N_1} + \vm{H}_{11} \bar{\vm{Q}}_{11} \vm{H}_{11}^\dagger - \vm{H}_{11}\bar{\vm{Q}}_{11}\vm{H}_{21}^\dagger (\vm{I}_{N_2} + \vm{H}_{21}\bar{\vm{Q}}_{11}\vm{H}_{21}^\dagger)^{-1}\vm{H}_{21}\bar{\vm{Q}}_{11}\vm{H}_{11}^\dagger)\big) \\
&\quad + \log\det (\vm{I}_{N_2} + \vm{H}_2 \bar{\vm{Q}} \vm{H}_2^\dagger)+ R_B^{[1,2]},\\
&\overset{(b)}{\leq}  \log\det (\vm{I}_{N_1} + P \vm{H}_{11}\vm{H}_{11}^\dagger - P^2 \vm{H}_{11}\vm{H}_{21}^\dagger (\vm{I}_{N_2} + P \vm{H}_{21}\vm{H}_{21}^\dagger)^{-1}\vm{H}_{21}\vm{H}_{11}^\dagger)\big) \\
&\quad + \log\det (\vm{I}_{N_2} + \vm{H}_2 \bar{\vm{Q}} \vm{H}_2^\dagger)+ R_B^{[1,2]}\\
&\overset{(c)}{\leq}  \log\det (\vm{I}_{N_1} + P \vm{H}_{11}\vm{H}_{11}^\dagger - P^2 \vm{H}_{11}\vm{H}_{21}^\dagger (\vm{I}_{N_2} + P \vm{H}_{21}\vm{H}_{21}^\dagger)^{-1}\vm{H}_{21}\vm{H}_{11}^\dagger)\big) \\
&\quad + \log\det (\vm{I}_{N_2} + 2P \vm{H}_2 \vm{H}_2^\dagger)+ R_B^{[1,2]}\\
&\overset{(d)}{=}  \log\det (\vm{I}_{N_1} + P \vm{H}_{11}\vm{H}_{11}^\dagger - P^2 \vm{H}_{11}\vm{H}_{21}^\dagger (\vm{I}_{N_2} + P \vm{H}_{21}\vm{H}_{21}^\dagger)^{-1}\vm{H}_{21}\vm{H}_{11}^\dagger)\big) \\
&\quad + \log\det (\vm{I}_{N_2} + 2P \vm{H}_{11} \vm{H}_{11}^\dagger + 2P \vm{H}_{22} \vm{H}_{22}^\dagger)+ R_B^{[1,2]}.
\end{split}
\end{align}
.

Here, (a) is due to Claim~\ref{claim:concavity}, (b) is due to Lemma~\ref{lemma:matrix_function}, Lemma~\ref{lemma:A_B} and \eqref{eq:Q_11}.
(c) is also due to Lemma~\ref{lemma:A_B} and \eqref{eq:Q} and finally (d) is due to definition of $\vm{H}_2$ at \eqref{eq:def_H}.
Eq. \eqref{eq:rate_final_1} is the result of combining what we have at \eqref{eq:ent_bound_1} and \eqref{eq:each_concave}.
With similar arguments, we can also derive \eqref{eq:rate_final_2}.

\end{proof}

We are now ready to prove Lemma~\ref{lemma:MIMO_dof_bound}.
\begin{proof}[Proof of Lemma~\ref{lemma:MIMO_dof_bound}]

Consider the bound at \eqref{eq:rate_final_1}, we have
\begin{align}
\log\det &\big( I_{N_1} + P\vm{H}_{11}\vm{H}_{11}^\dagger - P^2\vm{H}_{11}\vm{H}_{21}^\dagger \big( I_{N_2} + P\vm{H}_{21}\vm{H}_{21}^\dagger \big)^{-1} \vm{H}_{21}\vm{H}_{11}^\dagger \big)\nonumber\\
	    &\quad + \log\det \big(I_{N_2} + 2P\vm{H}_{22}\vm{H}_{22}^\dagger+2P\vm{H}_{21}\vm{H}_{21}^\dagger \big) + R_B^{[1,2]}\\
	    & \overset{(a)}{=} (\min\{N_1 , (M_1 - N_2)^+ \} + \min\{N_2, M_1+ M_2\}+ C_B^{12})\log P + o(\log P)
\end{align}
where, (a) is a consequence of the results in Lemma~\ref{lemma:dof_complex} and Lemma~\ref{lemma:dof_simple_sum}.
Dividing both sides by $\log P$ and letting $P\rightarrow \infty$ we get \eqref{eq:MIMO_dof_bound_1}.
Similarly, one can obtain \eqref{eq:MIMO_dof_bound_2}. 

\end{proof}

\section{Preliminaries on Graph Theory}\label{app:graph}

This Appendix provides some preliminaries and results of graph theory.
Let $G = (\mathcal{V},\mathcal{E})$ denote a graph $G$, with the vertex (node) set of $\mathcal{V}$ and the edge (link) set of $\mathcal{E}$.

\begin{definition}[Induced Subgraph]
For the graph $G = (\mathcal{V},\mathcal{E})$, let $\mathcal{S}\subseteq \mathcal{V}$ be arbitrary subset of vertices.
Then the induced subgraph $G[\mathcal{S}] = (\mathcal{S}, \mathcal{E}_s)$ is the graph whose vertex set is $\mathcal{S}$. 
Moreover, the set $\mathcal{E}_s \subseteq\mathcal{E}$ is the edge set of $G[\mathcal{S}]$ and consists of the edges in $\mathcal{E}$ with both ends in $\mathcal{S}$.
\end{definition}

\begin{definition}[Connected graph]
A connected graph is a graph in which for each arbitrary pair of vertices there is at least one path starting from one vertex and ending at the other.
For the graph $G = (\mathcal{V},\mathcal{E})$ to be connected, it is necessary to have $|\mathcal{E}|\geq |\mathcal{V}|-1$.
\end{definition}

\begin{definition}[Geodesic Distance]\label{def:distance}
In a graph $G = (\mathcal{V},\mathcal{E})$, the geodesic distance between two vertices $i$ and $j$, denoted by $d(i,j)$, is the number of edges in a shortest path from node $i$ to node $j$.
For every pair $i,j\in \mathcal{V}$ we have $d(i,j) = d(j,i)$ and $d(i,i) = 0$.
\end{definition}

\begin{definition}[Normalized Closeness Centrality]\label{def:closeness_centrality}
For a node $i\in \mathcal{V}$ of a graph $G = (\mathcal{V},\mathcal{E})$, the normalized closeness centrality, denoted by $C(i)$, is defined as the inverse of normalized average geodesic distance of node $i$ from all other nodes in the graph, i.e.,
\begin{align}
C(i) = \frac{|V|-1}{\sum_{j\in V} d(i,j)}.
\end{align}
The node with Maximum Closeness Centrality measure, has the shortest average distance from all other nodes of the graph.
\end{definition}

\begin{definition}[Neighboring set]\label{def:neigh}
$\mathcal{N}_{\mathcal{Q}}(\mathcal{S})$ indicates the set of vertices in the set $\mathcal{Q}$ which are neighbors of at least one of the nodes in the set $\mathcal{S}$.
\end{definition}

\begin{definition}[$\min$--cut]\label{def:min-cut}
Consider a directed graph $G$, and two of its nodes $s$ and $t$, called source and sink, respectively.
Then, $\min$--cut of $G$, denoted by $\min$--cut$(G)$, is defined as the set of edges (links), with the minimum sum weight, by removing which, there remains no path from $s$ to $t$.
\end{definition}

\begin{definition}[Connectivity degree]
For node $i\in \mathcal{V}$ of the graph $G = (\mathcal{V},\mathcal{E})$, the connectivity degree is defined as the number of neighbors of node $i$ and is denoted by $\text{Deg}(i)$.
Apparently,  $0 \leq \text{Deg}(i) \leq |\mathcal{V}|-1$ for all $i\in \mathcal{V}$.
\end{definition}

\begin{definition}[Bipartite Graph]
A graph $G = (\mathcal{V},\mathcal{E})$ is Bipartite if there exists partition $\mathcal{V} =\mathcal{X} \cup \mathcal{Y}$ with $\mathcal{X} \cap \mathcal{Y} = \emptyset$ and $\mathcal{E} \subseteq \mathcal{X} \times \mathcal{Y}$. 
\end{definition}

\begin{definition}[Matching]
In bipartite graphs, a matching is a set of edges that do not have a set of common vertices. In other words, a matching is a sub--graph where each node has either zero or one edge incident to it.
One of the bipartitions is said to be saturated in a matching, whenever all the vertices in that bipartition are incident to exactly one edge.
In a perfect matching, both bipartitions are saturated.
\end{definition}

\begin{definition}[Independent set]\label{def:IS}
An independent set of a bipartite graph $G$, with bipartitions $(\mathcal{A},\mathcal{B})$, is a subset of nodes with no links among them.
An independent set is said to be \emph{proper}, if it contains at least one node from $\mathcal{A}$ and one node from $\mathcal{B}$.
We use the shorthand PIS for a proper independent set.
For any $a\in\mathcal{A}$ and $b\in\mathcal{B}$, we define $\max$--PIS$(a,b)$ as a proper independent set  with maximum size that contains  $a$ and $b$.
\end{definition}

For bipartite graphs we have the well-known \textit{Hall's Marriage} theorem (\cite{diestel2000graph}, Theorem 2.1.2, page 31 ) as follows.
\begin{Theorem}\label{theorem:hall}
 Let $G$ be a bipartite graph with bipartitions $(\mathcal{A},\mathcal{B})$. 
Subsequently, $G$ contains a matching that saturates every vertex in $\mathcal{A}$ if and only if  for all $\mathcal{S} \subseteq \mathcal{A}
$
\begin{align}\label{cond:hall}
|N_\mathcal{B}(\mathcal{S})| \geq |\mathcal{S}|,
\end{align}
where $N_{\mathcal{B}}(\mathcal{S})$ denotes the set of neighboring nodes of $\mathcal{S}$ in $\mathcal{B}$.
\end{Theorem}

The condition in~\eqref{cond:hall} is called Hall's Condition and can also be stated as the following condition.

\begin{condition}[Hall's Condition]\label{cond:hall-cond}
For every $\ell\in\{1,2,\ldots,|\mathcal{A}|\}$ and any subset of $\mathcal{A}$ with $\ell$ nodes, the number of neighbors of that subset in $\mathcal{B}$ must be at least $\ell$.
\end{condition}

The following corollary is the direct consequence of Theorem~\ref{theorem:hall}.

\begin{corollary}\label{perfect}
For a bipartite graph with equal size bipartitions, a perfect matching exists if and only if \eqref{cond:hall} holds.
\end{corollary}

We also have the following lemma about rank efficiency of a matrix with random elements. 
\begin{lemma}\label{lemma:rank_efficiency}
Consider a squared matrix $\vm{X}$, where the elements are either zero or independently chosen from a continuous distribution.
If the diagonal elements are random, then $\vm{X}$ is full rank, almost surely.
\end{lemma}

\begin{proof}
For a full rank matrix we know that the determinant is non-zero.
The proof is based on induction.
For one by one matrix, the rank efficiency is obvious.
Let us now assume that the assumption is valid for $n-1$ by $n-1$ matrices and consider the case of $n$ by $n$.
Using the determinant expansion on the row $n$, we have
\begin{align}\label{eq:det_exp}
|\vm{X}| = \sum_{j=1}^{n} ( - X_{nj})^{n+j} M_{nj},
\end{align}
where $M_{nj}$ indicates $n$, $j$ minor of $\vm{X}$, that is, the determinant of the $n-1$ by $n-1$ matrix that results from deleting the $n$--th row and the $j$--th column of $\vm{X}$.
Rewriting what we have in \eqref{eq:det_exp}, we have
\begin{align}
|\vm{X}| = \sum_{j=1}^{n-1} ( - X_{nj})^{n+j} M_{nj} + ( X_{nn})^{2n} M_{nn}.
\end{align}

$X_{nn}$ is a random number and $M_{nn}$ is non-zero due to the induction hypothesis, therefore $|\vm{X}|$ is non-zero, almost surely.
\end{proof}

 Consider an $M$ by $N$ matrix $\vm{X}$ with all elements equal to zero, where we are able to replace each of the elements independently with a random number choosing from a continuous distribution.
The question is how many and which elements do we need to replace in order for the resulting matrix to be full rank.
This matrix has an equivalent bipartite graph with bipartitions $(\mathcal{A},\mathcal{B})$ such that $|\mathcal{A}| = M$ and $|\mathcal{B}|=N$.
For non-zero element $X_{ij}$, there exists a direct link between nodes $i\in\mathcal{A}$ and $j\in\mathcal{B}$.  
The following proposition, answers the above question.

\begin{proposition}\label{prop:1}
An $M$ by $N$ matrix $\vm{X}$ with all elements either zero or drawn i.i.d. from a continuous distribution, is rank efficient with probability one, if and only if there exists a matching in the equivalent bipartite graph that saturates all the elements of the bipartiotion with fewer number of nodes, i.e., $\min(M,N)$.
\end{proposition}

\begin{proof}
 Without loss of generality we assume that $\min(M,N) = M$.\\
\textbf{Sufficiency proof:}

In case such a matching exists, there are at least $M$ non-zero elements, each in a different row--column pair.
We form a new matrix $\bar{\vm{X}}$, by substituting rows and columns such that these non-zero elements form $\bar{X}_{ii}$ for $i=1,2,\ldots,M$.
This is equivalent to changing the place of nodes in bipartitions, such that each element in first bipartition becomes the neighbor of its opponent in the second bipartition.
Consequently, we have $\bar{\vm{X}} = [\vm{X}_1 , \vm{X}_2]$ where $\vm{X}_1$ is an $M$ by $M$ matrix with non-zero diagonal elements and $\vm{X}_2$ is an $M$ by $N-M$ matrix.
Due to the Lemma~\ref{lemma:rank_efficiency} in Appendix \ref{app:linear_algebra}, $\vm{X}_1$ is full--rank and therefore, $\bar{\vm{X}}$ is rank efficient irrespective of the value of $\vm{X}_2$.
Since row and column operations do not change the rank of a matrix, we have $\text{rank}(\vm{X}) =\text{rank}(\bar{\vm{X}})$.
\\
\textbf{Necessity proof:}

 Assume that such matching does not exist and the equivalent bipartite graph consists of bipartitions $(\mathcal{A},\mathcal{B})$.
Due to Theorem~\ref{theorem:hall} in Appendix \ref{app:graph}, there exists a group of nodes $\mathcal{S}\subseteq \mathcal{A}$ with $|\mathcal{S}| = l$ and $|N(\mathcal{S})|\ \textless\ l\ $.
By substituting rows and columns of $\vm{X}$, we form a new matrix $\bar{\vm{X}}$ such that the nodes in $\mathcal{S}$ form the first $l$ rows.
We further have $\bar{\vm{X}} = [\vm{X}_1^\h, \vm{X}_2^\h]^\h$ where $\vm{X}_1$ is an $l$ by $N$ matrix and $\vm{X}_2$ is an $M-l$ by $N$ matrix.
Since $|N(\mathcal{S})|\ \textless\ l\ $, there are less than $l$ non-zero columns in $\vm{X}_1$.
Therefore, $\text{rank}(\vm{X}_1)\ \textless\ l$ and $\text{rank}(\bar{\vm{X}})\ \textless\ M$ no matter what $\vm{X}_2$ is.
Since row and column operations do not change the rank of a matrix, we have $\text{rank}(\vm{X}) =\text{rank}(\bar{\vm{X}})$ and $\vm{X}$ is rank-deficient.

\end{proof}

\begin{proposition}\label{prop:2}
Consider the matrix $\vm{A}$ consisting of $m\times n$ blocks of the following form
\begin{align}
\left[\begin{array}{cccc}
L_{11}\vm{A}_{11} & L_{12}\vm{A}_{12} & \ldots & L_{1n}\vm{A}_{1n}\\
L_{21}\vm{A}_{21} & L_{22}\vm{A}_{22} & \ldots & L_{2n}\vm{A}_{1n}\\
\vdots & \vdots & \ddots & \vdots\\
L_{n1}\vm{A}_{m1} & L_{n2}\vm{A}_{m2} & \ldots & L_{mn}\vm{A}_{mn}
\end{array}\right]
\end{align}
where $L_{ij}$ is either zero or one and the block $\vm{A}_{ij}$ is an $M$ by $M$ matrix with all the elements drawn i.i.d. from a continuous distribution, for all $i\in\{1,2,\ldots, m\}$ and $j\in\{1,2,\ldots,n\}$.

$\vm{A}$ is full rank with probability one, if and only if there exists a matching in the equivalent bipartite graph of adjacency matrix $\vm{L}$, that saturates all the elements of the bipartiotion with fewer number of nodes, i.e., $\min(m,n)$.
\end{proposition}

\begin{proof}

Due to the randomness, $\vm{A}_{ij}$ is full rank for all $i\in\{1,2,\ldots, m\}$ and $j\in\{1,2,\ldots,n\}$.
 Without loss of generality we assume that $\min(m,n) = m$.\\
\textbf{Sufficiency proof:}

In case such a matching exists, there are at least $m$ non-zero elements, each in a different row--column pair of $\vm{L}$.
One is able to form a new matrix $\bar{\vm{L}}$, by substituting rows and columns such that these non-zero elements form the main diagonal of $\vm{L}$, i.e., $\bar{L}_{ii}$ for $i=1,2,\ldots,M$.
This is equivalent to changing the place of nodes in bipartitions, such that each element in first bipartition becomes the neighbor of its opponent in the second bipartition.
By doing the same substitutions on the blocks of $\vm{A}$, we form a new matrix $\bar{\vm{A}}$.

Now we have $\bar{\vm{A}} = [\vm{A}_1 , \vm{A}_2]$ where $\vm{A}_1$ consists of $m\times m$ blocks where the diagonal blocks are random and $\vm{A}_2$ consists of $m\times (n-m)$ blocks.
Due to the Lemma~\ref{lemma:rank_efficiency} in Appendix \ref{app:linear_algebra}, since all the diagonal elements of $\vm{A}_1$ are non-zero with probability one, $\vm{X}_1$ is full--rank and therefore, $\bar{\vm{X}}$ is full--rank irrespective of the value of $\vm{X}_2$.
Since row and column operations do not change the rank of a matrix, we have $\text{rank}(\vm{X}) =\text{rank}(\bar{\vm{X}})$.
\\
\textbf{Necessity proof:}

 Assume that such matching does not exist and the equivalent bipartite graph consists of bipartitions $(\mathcal{A},\mathcal{B})$.
Due to Theorem~\ref{theorem:hall} in Appendix \ref{app:graph}, there exists a group of nodes $\mathcal{S}\subseteq \mathcal{A}$ with $|\mathcal{S}| = l$ and $|\mathcal{N}_{\mathcal{B}}(\mathcal{S})|\ \textless\ l\ $.
By substituting rows and columns of $\vm{L}$, we are able to form a new matrix $\bar{\vm{L}}$ such that the nodes in $\mathcal{S}$ form the first $l$ rows.
By doing the same substitutions on the blocks of $\vm{A}$, we form a new matrix $\bar{\vm{A}} = [\vm{A}_1^\h, \vm{A}_2^\h]^\h$ where $\vm{A}_1$ consists of $l\times n$ blocks and $\vm{A}_2$ consists of $m-l\times n$ blocks.
Since $|\mathcal{N}_{\mathcal{B}}(\mathcal{S})|\ \textless\ l\ $, there are less than $l\times n$ non-zero blocks in $\vm{A}_1$.
Therefore, $\text{rank}(\vm{A}_1)\ \textless\ l$ and $\text{rank}(\bar{\vm{A}})\ \textless\ m$ irrespective of the value of $\vm{A}_2$.
Since row and column operations do not change the rank of a matrix, we have $\text{rank}(\vm{A}) =\text{rank}(\bar{\vm{A}})$ and $\vm{A}$ is rank-deficient.

\end{proof}

Here, considering the adjacency matrix of a $K$--user interference channel and its equivalent bipartite graph, we introduce the dual for \nameref{cond:ehc} (Condition~\ref{cond:ehc}).
\begin{condition}[Dual of Extended Hall's Condition]\label{cond:ehcr}
Let $k$ be any arbitrary integer in $\{1,2,\ldots,\ceil{\frac{K}{2}}\}$.
In the equivalent bipartite graph, for each group of $k$ arbitrary subset $\mathcal{Q}\subseteq\mathcal{R}$ of \emph{receivers},
\begin{align*}
|\mathcal{N}_\mathcal{T}(\mathcal{Q})| \geq \floor{\frac{K}{2}}+k.
\end{align*}
\end{condition}
The reason for the appellation of \nameref{cond:ehcr} (Condition~\ref{cond:ehcr}), is that it is exactly the same as \nameref{cond:ehc} (Condition~\ref{cond:ehc}).
However, \nameref{cond:ehc} deals with the transmitters in the equivalent bipartite graph, while \nameref{cond:ehcr} deals with the receivers.
We have the follwoing lemma.

\begin{lemma}\label{lemma:trans_rec}
\nameref{cond:ehc} (Condition~\ref{cond:ehc}) and \nameref{cond:ehcr} (Condition~\ref{cond:ehcr}) are equivalent.
\end{lemma}

\begin{proof}
Let $k$ be any arbitrary integer in $\{1,2,\ldots,\ceil{\frac{K}{2}}\}$ and assume there exist a set $\mathcal{Q}\subseteq \mathcal{R}$, containing $k$ receivers with less than $\floor{\frac{K}{2}}+k$ neighbors, i.e.,
\begin{align*}
|\mathcal{N}_\mathcal{T}(\mathcal{Q})|\ \textless\ \floor{\frac{K}{2}}+k.
\end{align*}
Subsequently, there exists a set $\mathcal{S}\subseteq \mathcal{T}$, containing $l\ \textgreater\ \ceil{\frac{K}{2}}-k$ transmitters, where
\begin{align*}
|\mathcal{N}_\mathcal{T}(\mathcal{Q}) \cap \mathcal{S}| = 0,
\end{align*}
and,
\begin{align}\label{eq:contradict}
|\mathcal{N}_\mathcal{R}(\mathcal{S})| \leq K-k.
\end{align}
If $l\ \textgreater\ \ceil{\frac{K}{2}}$, according to \nameref{cond:ehc} (Condition~\ref{cond:ehc}), we must have $|\mathcal{N}_\mathcal{R}(\mathcal{S})| = K$, which contradicts \eqref{eq:contradict}, since $k \geq 1$. 
On the other hand, using \nameref{cond:ehc} (Condition~\ref{cond:ehc}) for all $l\leq \ceil{\frac{K}{2}}$, we have 
\begin{align}\label{eq:contradict1}
|\mathcal{N}_\mathcal{R}(\mathcal{S})|\ \geq \floor{\frac{K}{2}}+l.
\end{align}
However, as stated before $l\ \textgreater\ \ceil{\frac{K}{2}}-k$, and combining it with \eqref{eq:contradict1} results in
\begin{align}
|\mathcal{N}_\mathcal{R}(\mathcal{S})|\ \textgreater\ K - k,
\end{align}
which contradicts \eqref{eq:contradict}.

Changing the role of receivers and transmitters, one is able to show that if the condition in the lemma holds, then \nameref{cond:ehc} (Condition~\ref{cond:ehc}) holds.

\end{proof}

\bibliographystyle{ieeetr}
\bibliography{journal_abbr,multiuserbh}

\end{document}